\documentclass[a4paper]{article}
\usepackage{geometry}
\usepackage[utf8]{inputenc}
\usepackage[T1]{fontenc}

\usepackage{graphicx,xcolor}
\usepackage{amsmath,amssymb,amsfonts}
\usepackage{url}
\usepackage{hyperref}
\usepackage{comment}
\usepackage{lscape}
\usepackage{viewplot}

\usepackage{amsthm}
\newtheorem{theorem}{Theorem}
\newtheorem{lemma}{Lemma}

\theoremstyle{definition}
\newtheorem{definition}{Definition}

\newcommand{\ra}{\footnotesize$\rightarrow$}
\newcommand{\subcaption}[1]{{\small #1}}

\title{Evacuation from Various Types of Finite 2D Square Grid Fields by a Metamorphic Robotic System}

\usepackage{authblk}

\author[1]{Junya Nakamura\thanks{Corresponding author: junya[at]imc.tut.ac.jp}}
\author[2]{Sayaka Kamei}
\author[3]{Yukiko Yamauchi}

\affil[1]{Toyohashi University of Technology, Japan}
\affil[2]{Hiroshima University, Japan}
\affil[3]{Kyushu University, Japan}

\date{}

\begin{document}

\maketitle

\begin{abstract}
A metamorphic robotic system (MRS) is composed of anonymous, memoryless, and autonomous modules that execute an identical distributed algorithm to move while keeping the connectivity of the modules.
For an MRS, the number of modules required to solve a given task is an important complexity measure.
Here, we consider evacuation from a finite two-dimensional square grid field by an MRS.
This study aims to establish the minimum number of modules required to solve the evacuation problem under several conditions.
We consider a rectangular field surrounded by walls with at least one exit.
Our results show that two modules are necessary and sufficient for evacuation from any rectangular field if equipped with a global compass, which provides the modules with a common sense of direction.
After that, we focus on the case of modules without a global compass and show that four (resp.~seven) modules are necessary and sufficient for restricted (resp.~any) initial shapes of an MRS.
We also show that two modules are sufficient when an MRS is touching a wall in an initial configuration.
Then, we clarify the condition to stop an MRS after evacuation of a rectangular field.
Finally, we extend these results to mazes and convex fields.
\end{abstract}

\section{Introduction}
Modular robotic systems, composed of a large number of autonomously-moving modules, are attracting much attention among both robotics~\cite{Tucci2018,Thalamy2019,Thalamy2019a} and theoretical computer science~\cite{Dumitrescu2004a,Dumitrescu2004,Chen2014,Doi2021} researchers.
Because the modules act collectively to complete a given task, research has been focused primarily on distributed coordination.
One of the most important distributed computing models is the \emph{metamorphic robotic system} (MRS)~\cite{Dumitrescu2004a,Dumitrescu2004}, that consists of a set of anonymous modules in a two-dimensional (2D) square grid.
Each cell of the grid can accept at most one module at each time step, and each module can perform two types of movements, i.e., \emph{rotation} and \emph{sliding}.
The former is a rotation by $\pi/2$ clockwise or counter-clockwise around an adjacent module, and the latter is a straight-line movement along the static modules aligned together in a straight line.
However, the modules must maintain \emph{connectivity} among each other, defined by side-adjacency of cells occupied by the modules.
The modules are \emph{anonymous}, \emph{uniform}, and \emph{oblivious}; they are indistinguishable, execute a common algorithm, and are memoryless.
The modules are synchronously activated in each time step, observe the positions of other modules, compute their movement, and perform the movement.
The MRS changes its shape and moves toward its destination using local movements of its modules.
Previous studies in the literature~\cite{Dumitrescu2004a,Dumitrescu2004,Chen2014,Doi2021} have investigated the computational power of the MRS with a variety of problems and evaluation criteria.

Table~\ref{tab:contribution-summary} summarizes the existing results of the metamorphic robotic systems.
\begin{landscape}
\begin{table}[t]
	\centering
	\caption{Summary of related work and our contribution. %
	RF and CF are the abbreviations of rectangular fields and convex fields, respectively.}
	\label{tab:contribution-summary}
	\footnotesize

\begin{tabular}{lccccccc}
\hline
Problem & Algorithm & Field shape & Global compass & Visibility & \#modules & Initial states & Stop after execution\\
\hline \hline
Shape formation & \cite{Dumitrescu2004} & -- & Yes & Unlimited & $\geq 2$ & Convex shapes & Yes \\
\hline
Locomotion	& \cite{Dumitrescu2004a}	& --	& Yes	& Unlimited		& $\geq 2$	& Any & -- \\
			& \cite{Chen2014}			& --	& Yes	& $7 \times 7$	& $2$	& Convex shapes & -- \\
			& \cite{Doi2021}			& --	& No	& $7 \times 7$	& $4$	& Restricted (4 forbidden states) & -- \\
			& \cite{Doi2021}			& --	& No	& $11 \times 11$& $7$  & Any & -- \\
\hline
Search	& \cite{Doi2021}	& RF	& Yes	& $5 \times 5$		& $3$	& Any & Yes\\
		& \cite{Doi2021}	& RF	& No	& $9 \times 9$		& $5$	& Restricted (4 forbidden states) & Yes\\
		& \cite{Doi2021}	& RF	& No	& $11 \times 11$	& $7$	& Any & Yes\\
\hline
Evacuation	& Sect.~\ref{sec:evacuation-rectangular-with-global-compass}	& RF	& Yes		& $5 \times 5$		& 2		& Any & Depends on \#exits \\
			& Sect.~\ref{sec:evacuation-without-gc-restricted-init-states}	& RF	& No		& $7 \times 7$		& 4		& Restricted (4 forbidden states) & Yes \\
			& Sect.~\ref{sec:evacuation-without-gc-arbitrary-init-states}	& RF	& No		& $11 \times 11$	& 7		& Any & Depends on \#exits \\
			& Sect.~\ref{sec:evacuation-without-gc-along-the-wall}			& RF	& No		& $5 \times 5$		& 2		& Touching a wall & Depends on \#exits\\
			& Sect.~\ref{sec:evacuation-from-maze}							& Maze	& Yes/No/No	& $5/7/11 \times 5/7/11$	& 2/4/7	& Same to the results of RF	& No \\
			& Sect.~\ref{sec:evacuation-from-convex-field}					& CF	& Yes/No/No	& $5/7/11 \times 5/7/11$	& 2/4/7	& Same to the results of RF & Yes/No \\
\hline
\end{tabular}
\end{table}
\end{landscape}
Dumitrescu et al. studied \emph{reconfiguration} of the MRS to change its initial shape to a given target shape~\cite{Dumitrescu2004}.
They proved that any horizontally convex shape can be translated into a horizontal chain shape, i.e., a line.
An MRS shape is said to be \emph{horizontally convex} if the set of modules in the same row is connected.
Thus, any horizontally convex shape can be translated into another horizontally convex shape through the chain shape utilizing the reversibility of local movements.
They also consider decidability of related reconfiguration problems.

Dumitrescu et al. considered the \emph{locomotion} of the MRS to a given direction~\cite{Dumitrescu2004a} and demonstrated how the MRS achieves the fastest locomotion to a given vertical and diagonal direction.
Chen et al. proposed a locomotion algorithm for modules that can observe cells within a constant distance~\cite{Chen2014}.

Doi et al. proposed an algorithm for the \emph{search} problem by the MRS~\cite{Doi2021}.
The search problem requires an MRS to reach a target cell in an unknown finite 2D rectangular field from an arbitrary initial configuration.
The authors focused on the number of modules necessary and sufficient for the MRS to achieve search because the shape of anonymous modules serves as a memory of the MRS; for example, modules can store a direction of motion and progress of search.
They also demonstrated the effect of the \emph{global compass}; when the modules agree on the directions north, south, east, and west, we say they are equipped with a global compass.
They showed that three modules are necessary and sufficient from an arbitrary initial shape when the modules are equipped with a global compass.
Thus, the MRS guarantees \emph{self-stabilizing}~\cite{Dijkstra1974} search.
When the modules are not equipped with a global compass, modules cannot perform any movement in some initial shapes due to their symmetry.
For example, when the four modules form a square, none of them can move because symmetric modules only perform symmetric movements, and no module can serve as a static module for movement.
They showed that when the modules are not equipped with a global compass but share a common handedness, five modules are necessary and sufficient starting from restricted initial shapes.
They also showed that seven modules with the same ability are necessary and sufficient starting from arbitrary initial shapes.

\noindent{\bf Our contribution.~}
The last six rows of Table \ref{tab:contribution-summary} summarizes our contribution.
We first consider \emph{evacuation} by the MRS from an unknown finite 2D rectangular field.
In this case, the field is surrounded by walls with at least one exit on the walls.
The MRS is required to exit from the field without a priori knowledge of the field, exit position, or its initial position.
To the best of our knowledge, this is the first time study of evacuation by the MRS.
We show that when the modules are equipped with a global compass, two modules are necessary and sufficient for evacuation from an arbitrary initial shape, and when the modules are not equipped with a global compass but share a common handedness, four modules are necessary and sufficient from restricted initial shapes.
Then, we show that under the latter conditions, seven modules are necessary and sufficient for evacuation from an arbitrary initial shape, and two modules are necessary and sufficient for evacuation from a configuration where the MRS is initially touching a wall.
This result separates the problem of moving to a wall from the problem of searching for an exit on the wall.
Finally, we demonstrate that our evacuation algorithms can be extended to evacuation from an unknown maze composed of rectangular fields and from a convex field.

We also clarify the condition to stop an MRS from moving after evacuation.
In the search problem, an MRS can detect completing a given task easily and stops moving because a target cell is distinguishable from other types of cells.
However, in the evacuation problem, the detection and the stop are difficult because a field has no clear sign that helps to detect the completion.
Moreover, in most cases, an MRS cannot distinguish by observing its neighbor cells whether it is inside a field or outside a field.
We prove that an MRS can stop moving if a rectangular field has only one exit or restrictions on its initial shape.

\noindent\textbf{Related work.}~
A variety of theoretical computational models for modular robotic systems have been proposed in the field of distributed computing theory.
The \emph{mobile robot system}, also known as the \emph{LCM robot system}, considers a set of anonymous mobile robots moving in 2D Euclidean space~\cite{Suzuki1999}.
The robots are anonymous, uniform, communication-less, and have no access to the global coordinate system.
Each robot repeats a cycle, in which it observes the positions of other robots and computes and moves to its next position.
A robot is \emph{oblivious} if the input to computation is the preceding observation, otherwise \emph{non-oblivious}.
Di Luna et al. showed that a constant number of oblivious mobile robots can simulate a single non-oblivious mobile robot by encoding the contents of the local memory of a non-oblivious robot to their geometric positions~\cite{DiLuna2018}.

The \emph{programmable particle model} consists of anonymous particles with constant size local memory in an infinite 2D triangular grid~\cite{Derakhshandeh2014}.
Each vertex of the triangular grid can accept at most one particle at each time step. 
Each particle can detect whether a neighboring vertex is occupied by another particle and can communicate with particles at neighboring vertices.
Each particle moves by repeating \emph{expansion} and \emph{contraction}; a particle is expanded when it occupies two neighboring vertices and contracted when it occupies a single vertex.
Daymude et al. showed that programmable particles collectively form a global counter by forming a chain shape, and they can form a convex hull of an object of unknown size~\cite{Daymude2020}.
Di Luna et al. considered \emph{shape formation} by programmable particles and showed that the programmable particles can simulate a Turing machine~\cite{DiLuna2020} and RAM~\cite{DiLuna2020b}.

Cooperative evacuation has been considered for mobile robots initially placed at the center of a disk with an exit on its boundary~\cite{Czyzowicz2014}.
The goal is to minimize the time required for all the robots to exit.
Contrary to the mobile robot system that we mentioned above, mobile robots considered in the paper are equipped with two communication medias: \emph{wireless communication} and \emph{local communication}.
They presented evacuation algorithms and lower bounds of evacuation time for two and three robots.
Czyzowicz et al. considered evacuation of three robots equipped with wireless communication
when at least one of them is faulty~\cite{Czyzowicz2017}, and evacuation of at most four robots when one of the robots is the ``queen'' and the goal is to guide the queen to the exit~\cite{Czyzowicz2020}.
Evacuation has many applications such as exploration, rescue, map construction, and navigation in dangerous terrain, disaster area, etc.

All these results are expected to serve as theoretical foundations for many related areas, such as robotics, navigation, autonomous vehicles, nano-manufacturing, molecular robotics, and understanding natural systems.

\noindent\textbf{Organization.}~
This paper is organized as follows:
Section \ref{sec:preliminaries} defines our model and the problems considered here.
First, we consider the evacuation problem from a rectangular field in Section \ref{sec:evacuation-rectangular}.
In Sections \ref{sec:evacuation-rectangular-with-global-compass} and \ref{sec:evacuation-rectangular-wo-global-compass}, we discuss evacuation by an MRS composed of modules equipped with a global compass and modules not equipped with a global compass, respectively.
In Section \ref{sec:stop-after-evacuation}, we clarify the condition to stop an MRS from moving after evacuation of a rectangular field.
Finally, in Section \ref{sec:other-types-of-field}, we extend these results to other field types: mazes and convex fields.
We conclude this paper in Section \ref{sec:conclusion}.

\section{Preliminaries}
\label{sec:preliminaries}

\subsection{Model}

We consider the rectangular MRS introduced in \cite{Chen2014,Doi2021,Dumitrescu2004,Dumitrescu2004a} and a 2D square grid where each square cell $c_{i,j}$ is labeled by the underlying $x$-$y$ coordinate system, as shown in Fig.~\ref{fig:cell}.
\begin{figure}[t]
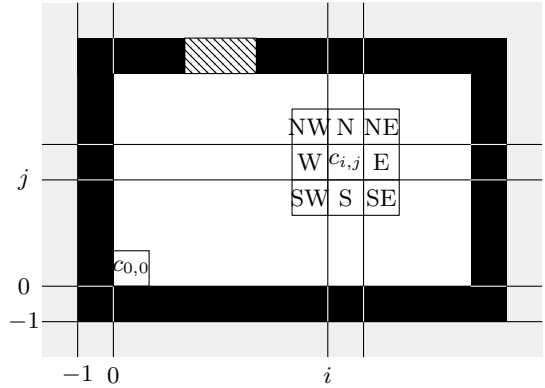

	\centering
	\setlength{\gridstep}{4.7mm}
\begin{nogridplot}{14}{10}

	\outside{1}
	\gwall{1}{1}{12}{1}
	\gwall{1}{1}{1}{8}
	\gwall{12}{1}{12}{8}
	\gwall{1}{8}{3}{8}
	\gwall{6}{8}{12}{8}
	\gexit{4}{8}{5}{8}

	\draw [black] (0.0,1.0) -- (14.0,1.0);
	\draw [black] (0.0,2.0) -- (14.0,2.0);
	\draw [white] (1.0,2.0) -- (2.0,2.0);
	\draw [white] (12.0,2.0) -- (13.0,2.0);

	\draw [black] (0.0,5.0) -- (14.0,5.0);
	\draw [white] (1.0,5.0) -- (2.0,5.0);
	\draw [white] (12.0,5.0) -- (13.0,5.0);

	\draw [black] (0.0,6.0) -- (14.0,6.0);
	\draw [white] (1.0,6.0) -- (2.0,6.0);
	\draw [white] (12.0,6.0) -- (13.0,6.0);

	\draw [black] (1.0,0.0) -- (1.0,10.0);
	\draw [black] (2.0,0.0) -- (2.0,10.0);
	\draw [white] (2.0,1.0) -- (2.0,2.0);
	\draw [white] (2.0,8.0) -- (2.0,9.0);

	\draw [black] (8.0,0.0) -- (8.0,10.0);
	\draw [white] (8.0,1.0) -- (8.0,2.0);
	\draw [white] (8.0,8.0) -- (8.0,9.0);

	\draw [black] (9.0,0.0) -- (9.0,10.0);
	\draw [white] (9.0,1.0) -- (9.0,2.0);
	\draw [white] (9.0,8.0) -- (9.0,9.0);

	\draw [step=\the\gridstep] (7.0,4.0) grid (10.0,7.0);
	\draw (2.0,2.0) rectangle (3.0,3.0);

	\node at (2.5, 2.5) {\small $c_{0,0}$};

	\node at (8.5, 6.5) {\small N};
	\node at (8.5, 5.5) {\small $c_{i,j}$};
	\node at (8.5, 4.5) {\small S};

	\node at (9.5, 6.5) {\small NE};
	\node at (9.5, 5.5) {\small E};
	\node at (9.5, 4.5) {\small SE};

	\node at (7.5, 6.5) {\small NW};
	\node at (7.5, 5.5) {\small W};
	\node at (7.5, 4.5) {\small SW};

	\node at (8.0, -0.5) {\small $i$};
	\node at (-0.5, 5.0) {\small $j$};

	\node at (2.0, -0.5) {\small 0};
	\node at (-0.5, 2.0) {\small 0};

	\node at (1.0, -0.5) {\small $-1$};
	\node at (-0.5, 1.0) {\small $-1$};

\end{nogridplot}
	\caption{Cell $c_{i,j}$ and its eight adjacent cells in a field.
		The interior (resp. exterior) is represented by white (resp. gray) cells.
		The exit is filled with diagonal lines.}
	\label{fig:cell}
\end{figure}
Each cell $c_{i,j}$ has eight \emph{adjacent} cells: East (E) $c_{i+1,j}$, NorthEast (NE) $c_{i+1,j+1}$, North (N) $c_{i,j+1}$, NorthWest (NW) $c_{i-1,j+1}$, West (W) $c_{i-1,j}$, SouthWest (SW) $c_{i-1,j-1}$, South (S) $c_{i,j-1}$, and SouthEast (SE) $c_{i+1,j-1}$.
Also, we say that the four cells, N, S, E, and W, are \emph{side-adjacent} to $c_{i,j}$.
An infinite sequence of cells with the same $x$ (resp.~$y$) coordinate is called a \emph{column} (resp.~a \emph{row}).
A \emph{field} is a rectangular subgrid in the square 2D grid, divided into \emph{interior} and \emph{exterior} regions.
The interior is surrounded by walls that separate it from the exterior.
A wall is composed of cells in the same row or column, and a module cannot occupy a cell of a wall.
We assume that the exterior is sufficiently large relative to the number of modules.
Without loss of generality, we assume that $c_{0,0}$ (resp. $c_{w-1,h-1}$) is the southwesternmost cell (resp. the northeasternmost cell) of the interior.
The interior has at least one \emph{exit} to the exterior.
The exit is on a wall, and it consists of a sequence of cells.
The interior is connected to the exterior only by an exit.
These cell labels are used just for description, and there is no way to distinguish cells.
We consider various types of fields and detail them in Section \ref{sec:fields}.

An MRS, $R$ consists of $n$ $(\geq 2)$ anonymous modules, each of which occupies a distinct cell in the square grid at discrete time step $t=0,1,2,\cdots$.
The \emph{configuration} $C_t$ of $R$ at time $t$ is the set of cells occupied by the modules at time $t$.
An \emph{execution} is an evolution of configurations $C_0,C_1,C_2,\cdots$.
The evolution is generated by movements of modules.
Let $M_t$ be a set of modules that move at time $t$.
We call a set $B_t$ of the modules that do not move at time $t$ (that is, $B_t = C_t \setminus M_t)$ a \emph{backbone}.
Modules have two types of movement, \emph{rotation} and \emph{sliding}, which are both guided by backbone modules as shown in Fig.~\ref{fig:basic-move}.
\begin{figure}[t]
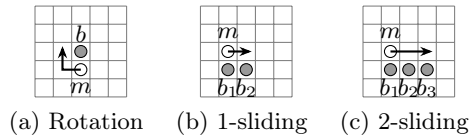

	\centering
	\begin{tabular}{ccc}
	\begin{viewplot}{5}
		\backbone{0}{0}
		\module{0}{-1}
		\draw[myarrow] (2.5,1.5) -- (1.5,1.5) -- (1.5,2.8);
		\node at (2.5, 0.5) {{\small $m$}};
		\node at (2.5, 3.5) {{\small $b$}};
	\end{viewplot}
	&
	\begin{viewplot}{5}
		\module{-1}{0}
		\backbone{-1}{-1}
		\backbone{0}{-1}
		\draw[myarrow] (1.5,2.5) -- (2.8,2.5);
		\node at (1.5, 0.5) {{\small $b_1$}};
		\node at (2.5, 0.5) {{\small $b_2$}};
		\node at (1.5, 3.5) {{\small $m$}};
	\end{viewplot}
	&
	\begin{viewplot}{5}
		\module{-1}{0}
		\backbone{-1}{-1}
		\backbone{0}{-1}
		\backbone{1}{-1}
		\draw[myarrow] (1.5,2.5) -- (3.8,2.5);
		\node at (1.5, 0.5) {{\small $b_1$}};
		\node at (2.5, 0.5) {{\small $b_2$}};
		\node at (3.5, 0.5) {{\small $b_3$}};
		\node at (1.5, 3.5) {{\small $m$}};
	\end{viewplot}
	\\
	\subcaption{(a) Rotation} & \subcaption{(b) 1-sliding} & \subcaption{(c) 2-sliding}
\end{tabular}

	\caption{Movements of a module}
	\label{fig:basic-move}
\end{figure}
During a rotation, a module $m$ moves either clockwise or counter-clockwise around a side-adjacent backbone module $b$ by $\pi/2$ (Fig.~\ref{fig:basic-move}(a)).
We classify sliding movements by their distance of motion.
By a 1-sliding, a module $m$ moves to a side-adjacent cell along two backbone modules, $b_1$ and $b_2$ (Fig.~\ref{fig:basic-move}(b)).
For this movement to occur, the backbone module $b_1$ (resp.~$b_2$) must be side-adjacent to $m$ (resp.~$b_1$).
Similarly to a 1-sliding, by a \emph{$k$-sliding} ($k \geq 2,3,\cdots$), a module $m$ moves $k$ side-adjacent cells along $k+1$ backbone modules.
For these movements, the cells that $m$ passes through during a movement must not be occupied by any modules.

The connectivity of configuration $C_t$ is represented by a connectivity graph $G_t = (C_t, E_t)$.
The edge set $E_t$ contains an edge $(c,c')$ for $c,c' \in C_t$ if and only if cells $c$ and $c'$ are side-adjacent.
If a connectivity graph $G_t$ induced by a configuration $C_t$ is connected, we say a configuration $C_t$ is \emph{connected}.
Any execution $C_0,C_1,C_2,\cdots$ of an MRS must satisfy the following three conditions:
\begin{itemize}
	\item \textbf{Connectivity}: for any time $t$, $C_t$ is connected.
	\item \textbf{Single backbone}: for any time $t$, $B_t$ is connected.
	\item \textbf{No interference}: for any time $t$, the trajectories of any two moving modules never overlap.
\end{itemize}

Modules of an MRS are \emph{uniform}; they are anonymous and execute an identical deterministic distributed algorithm.
Modules are \emph{oblivious}, that is, they have no memory.
Modules behave synchronously; at each time step, each module observes the cells in its neighborhood and decides its movement.
We say that a cell $c_{i',j'}$ is a \emph{$k$-neighborhood} of cell $c_{i,j}$ if $|i' - i| \leq k$ and $|j'-j| \leq k$.
We assume that a module can determine whether each cell in its $k$-neighborhood is occupied by a module or a wall and that $k$ is constant with respect to the size of a field.
A \emph{view} of a module $m$ is a $(2k+1) \times (2k+1)$ square subgrid centered at the cell that $m$ occupies.
A distributed algorithm of neighborhood size $k$ is defined by a total function that maps a view to the cell to which a module moves in this step. 
We assume that a module knows $n$, and an algorithm is not universal; an algorithm is designed for an MRS composed of exactly $n$ modules.

If the modules are equipped with a \emph{global compass}, they share common North, South, East, and West directions.
Otherwise, they do not know directions and their observations may be inconsistent.
However, we assume that the modules agree on the clockwise direction; they share a common handedness.

The \emph{state} of an MRS $R$ in $C_t$ is the local shape of $R$.
We denote by $S^n$ a state of $R$ composed of $n$ modules.
If the modules are equipped with a global compass, a state of $R$ contains global directions.
Otherwise, it does not contain any direction because the modules cannot coordinate any rotation of their state in the absence of a global compass.

If the modules are equipped with a global compass, the execution of a given algorithm is uniquely determined by $C_0$ because the modules can agree on a total ordering among themselves.
In contrast, if the modules are not equipped with a global compass, there exist multiple executions from $C_0$ depending on the local compass of each module.

\subsection{Problems}
\label{sec:problems}

Here, we consider the following evacuation problem.
\begin{definition}[Evacuation]
\label{def:evacuation}
A metamorphic robotic system has to evacuate from the interior of a field starting from any initial position and initial shape through an exit.
\end{definition}
\noindent We refer to evacuation from the interior to the exterior through an exit on a field as ``the evacuation from the field''.

As a building block of the proposed evacuation algorithms, we use existing algorithms \cite{Chen2014,Doi2021} that solve the following locomotion problem.
\begin{definition}[Locomotion]
\label{def:locomotion}
A metamorphic robotic system in an infinite 2D square grid is required to keep on moving in one direction.
\end{definition}
\noindent Hence, to solve the locomotion problem, modules of an MRS must break their initial symmetry and agree on a direction of motion.

\subsection{Fields}
\label{sec:fields}

In this section, we will describe the types of fields of the evacuation problem precisely.
First, we prove the following lemma regarding the basic movements of an MRS.
\begin{lemma}
\label{lem:exit-size}
An exit of a field must consist of at least two cells.
\end{lemma}
\begin{proof}
This can be proven using the fact that every basic move shown in Fig.~\ref{fig:basic-move} requires two or more cells that are not occupied by walls, as indicated in Fig.~\ref{fig:exit-size-limitation}.
Note that modules of an MRS must keep their connectivity during a move, and it is impossible for a module to pass through an exit which is only one cell wide without breaking connectivity.

For contradiction, we assume that there is an exit composed of one cell.
To pass the exit by a rotation, there must be a cell with a backbone module next to the exit (Fig.~\ref{fig:exit-size-limitation}(a)).
However, this is impossible because the cell is occupied by the wall.
Similarly, passing the exit by a $k$-sliding is also impossible because doing so also requires a backbone module in the wall (Fig. \ref{fig:exit-size-limitation}(b)).
On the other hand, if the exit consists of two or more cells, the backbone module does not need to be located in a cell that is part of a wall; thus, a module can pass through the exit by a rotation or a $k$-sliding.
\begin{figure}[t]
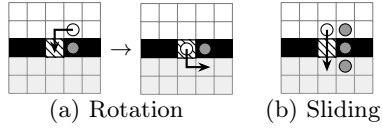

	\centering
	\begin{tabular}{cc}
	\begin{customviewplot}{5}{5}
		\gwall{0}{2}{1}{2}
		\gwall{3}{2}{4}{2}
		\outsidearea{0}{0}{4}{1}
		\exit{0}{0}{0}{0}

		\gmodule{3}{3}
		\gbackbone{3}{2}

		\draw[myarrow] (3.5,3.5) -- (2.5,3.5) -- (2.5,2.2);
	\end{customviewplot}
	\ra
	\begin{customviewplot}{5}{5}
		\gwall{0}{2}{1}{2}
		\gwall{3}{2}{4}{2}
		\outsidearea{0}{0}{4}{1}
		\exit{0}{0}{0}{0}

		\gmodule{2}{2}
		\gbackbone{3}{2}

		\draw[myarrow] (2.5,2.5) -- (2.5,1.5) -- (3.8,1.5);
	\end{customviewplot}
	&
	\begin{customviewplot}{5}{5}
		\gwall{0}{2}{1}{2}
		\gwall{3}{2}{4}{2}
		\outsidearea{0}{0}{4}{1}
		\exit{0}{0}{0}{0}

		\gmodule{2}{3}
		\gbackbone{3}{1}
		\gbackbone{3}{2}
		\gbackbone{3}{3}

		\draw[myarrow] (2.5,3.5) -- (2.5,1.2);
	\end{customviewplot}
	\\
	\subcaption{(a) Rotation} & \subcaption{(b) Sliding}\\
\end{tabular}

	\caption{Counterexamples of an exit composed of one cell}
	\label{fig:exit-size-limitation}
\end{figure}
\end{proof}

We consider the evacuation problem in rectangular fields, mazes, and convex fields.
From Lemma \ref{lem:exit-size}, we assume that every exit consists of two or more side-adjacent cells.
Figures \ref{fig:example-rectangular-fields}--\ref{fig:example-convex-fields} illustrate the examples of these types of field.

\begin{definition}
\label{def:rectangular-field}
A \emph{rectangular field} $F$ has a rectangular interior whose width and height are $w$ and $h$ cells, respectively ($w,h \geq 2$).
The interior of $F$ is surrounded by four walls each with a thickness of one cell, located at its North, South, East, and West.
The walls have at least one exit.
\end{definition}

\begin{definition}
\label{def:maze}
A \emph{maze} $M$ is composed of rectangular fields.
The walls of a rectangular field in $M$ are shared with other rectangular fields in $M$.
An exit of a rectangular field in $M$ connects to an exit of another rectangular field or the exterior.
The maze $M$ has at least one exit to the exterior.
\end{definition}

\begin{definition}
\label{def:convex-field}
We say that a set of connected cells is \emph{horizontally convex} (resp.~\emph{vertically convex}) when the cells at the same row (resp.~column) are connected for all the rows (resp.~columns).
A \emph{horizontally convex field} (resp.~\emph{vertically convex field}) has a horizontally (resp.~vertically) convex interior.
The interior is surrounded by walls whose thickness is one cell.
The walls have at least one exit.
\end{definition}
\noindent
The two types are rotationally symmetric; hence the two fields being different is inconsequential to the solvability of the evacuation problem.
Thus, we do not distinguish between these types and refer them both as \emph{convex fields}.

\begin{figure}[t]
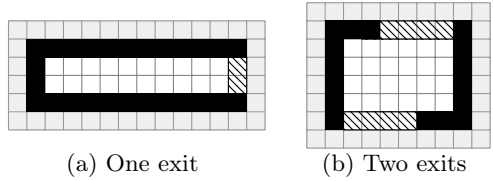

	\centering
	\begin{tabular}{cc}
\begin{customviewplot}{14}{6}
\outside{1}
\gwall{1}{1}{1}{4}
\gwall{1}{1}{12}{1}
\gwall{1}{4}{12}{4}
\gexit{12}{2}{12}{3}
\end{customviewplot}
&
\begin{customviewplot}{10}{8}
\outside{1}
\gwall{1}{1}{1}{6}
\gwall{8}{1}{8}{6}
\gwall{6}{1}{8}{1}
\gwall{1}{6}{3}{6}

\gexit{2}{1}{5}{1}
\gexit{3}{6}{7}{6}

\end{customviewplot}
\\
\subcaption{(a) One exit} & \subcaption{(b) Two exits}\\
\end{tabular}

	\caption{Examples of rectangular fields}
	\label{fig:example-rectangular-fields}
\end{figure}
\begin{figure}[t]
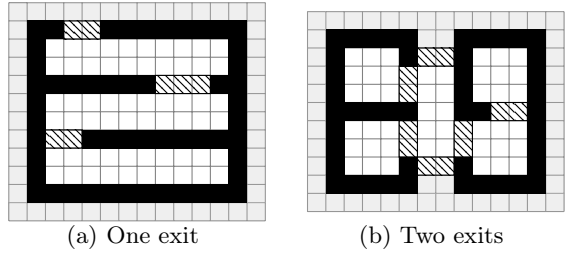

	\centering
	\begin{tabular}{cc}
\begin{customviewplot}{14}{12}
\outside{1}
\gwall{1}{1}{1}{10}
\gwall{12}{1}{12}{10}
\gwall{1}{1}{12}{1}
\gwall{4}{4}{12}{4}
\gwall{1}{7}{7}{7}
\gwall{11}{7}{12}{7}
\gwall{1}{10}{2}{10}
\gwall{5}{10}{12}{10}

\gexit{2}{4}{3}{4}
\gexit{8}{7}{10}{7}

\gexit{3}{10}{4}{10}

\end{customviewplot}
&
\begin{customviewplot}{14}{11}
\outside{1}
\outsidearea{5}{1}{7}{1}
\outsidearea{6}{9}{8}{9}

\gwall{1}{9}{5}{9}
\gwall{1}{5}{5}{5}
\gwall{1}{5}{1}{9}

\gwall{1}{1}{5}{1}
\gwall{1}{1}{1}{5}

\gwall{8}{2}{8}{2}
\gwall{5}{8}{5}{8}
\gwall{5}{2}{5}{2}
\gwall{8}{5}{8}{8}

\gwall{8}{9}{12}{9}
\gwall{12}{5}{12}{9}

\gwall{8}{5}{9}{5}
\gwall{8}{1}{12}{1}
\gwall{12}{1}{12}{5}

\gexit{5}{3}{5}{4}
\gexit{5}{6}{5}{7}

\gexit{8}{3}{8}{4}
\gexit{10}{5}{11}{5}

\gexit{6}{2}{7}{2}
\gexit{6}{8}{7}{8}

\end{customviewplot}
\\
\subcaption{(a) One exit} & \subcaption{(b) Two exits}\\
\end{tabular}

	\caption{Examples of mazes}
	\label{fig:example-mazes}
\end{figure}
\begin{figure}[t]
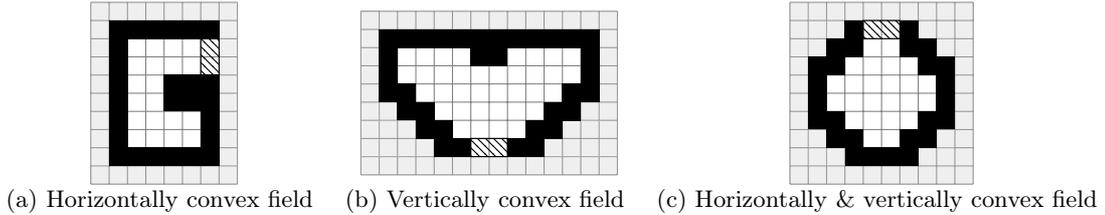

	\centering
	\begin{tabular}{ccc}
\begin{customviewplot}{8}{10}
\outside{1}
\gwall{1}{1}{1}{8}
\gwall{1}{1}{6}{1}
\gwall{1}{8}{6}{8}
\gwall{6}{1}{6}{5}
\gwall{4}{4}{5}{5}

\gexit{6}{6}{6}{7}
\end{customviewplot}
&
\begin{customviewplot}{14}{9}
\outside{1}
\outsidearea{1}{1}{2}{2}
\outsidearea{1}{3}{1}{3}
\outsidearea{3}{1}{3}{1}

\outsidearea{11}{1}{12}{2}
\outsidearea{10}{1}{10}{2}
\outsidearea{12}{3}{12}{3}

\gwall{1}{7}{12}{7}
\gwall{1}{4}{1}{7}
\gwall{1}{4}{2}{4}
\gwall{2}{3}{2}{4}
\gwall{3}{2}{3}{3}
\gwall{3}{2}{4}{2}
\gwall{4}{1}{4}{2}
\gwall{4}{1}{5}{1}
\gwall{8}{1}{9}{1}
\gwall{9}{1}{9}{2}
\gwall{9}{2}{10}{2}
\gwall{10}{2}{10}{3}
\gwall{10}{3}{11}{3}
\gwall{11}{3}{11}{4}
\gwall{11}{4}{12}{4}
\gwall{12}{4}{12}{7}
\gwall{6}{6}{7}{6}

\gexit{6}{1}{7}{1}

\end{customviewplot}
&
\begin{customviewplot}{10}{10}
\outside{1}
\outsidearea{1}{1}{1}{2}
\outsidearea{2}{1}{2}{1}

\outsidearea{1}{7}{1}{8}
\outsidearea{2}{8}{2}{8}

\outsidearea{7}{1}{8}{1}
\outsidearea{8}{2}{8}{2}

\outsidearea{7}{8}{8}{8}
\outsidearea{8}{7}{8}{7}

\gwall{3}{1}{6}{1}
\gwall{1}{3}{1}{6}
\gwall{8}{3}{8}{6}
\gwall{3}{8}{3}{8}
\gwall{6}{8}{6}{8}

\gwall{2}{2}{2}{3}
\gwall{2}{2}{3}{2}

\gwall{2}{6}{2}{7}
\gwall{2}{7}{3}{7}

\gwall{6}{2}{7}{2}
\gwall{7}{2}{7}{3}

\gwall{6}{7}{7}{7}
\gwall{7}{6}{7}{7}

\gexit{4}{8}{5}{8}

\end{customviewplot}
\\
\subcaption{(a) Horizontally convex field} &
\subcaption{(b) Vertically convex field} &
\subcaption{(c) Horizontally \& vertically convex field} \\
\end{tabular}

	\caption{Examples of convex fields}
	\label{fig:example-convex-fields}
\end{figure}

We represent a \emph{side-adjacency graph} of a set $V_c$ of cells by $G(V_c) = (V_c,E_c)$ where an edge $e=\{c,c'\}$ is in $E_c$ if and only if cells $c$ and $c'$ in $V_c$ are side-adjacent.
For a field $F$, we denote the set of all the wall cells by $W_F$, the set of all exit cells that are side-adjacent to the exterior by $X_F$, and the set of all the cells in the interior and all exit cells that connect the interiors of two rectangular fields\footnote{This type of exit cells occurs only in a maze.} by $I_F$.

To guarantee that evacuation from a maze is solvable, we assume that every maze $M$ considered here satisfies the following conditions:
\begin{enumerate}
    \item In the side-adjacency graph $G(I_M)$, every cell $c$ in $I_M$ has at least one path to a cell that is side-adjacent to an exit cell in $X_M$.
    \item The side-adjacency graph $G(W_M \cup X_M)$ is connected.
\end{enumerate}
Figure \ref{fig:maze-conditions} shows examples of mazes that do not satisfy one of these conditions.
Condition 1 guarantees the existence of an evacuation path for every cell in the interior.
The maze in Fig.~\ref{fig:maze-conditions}(a) does not satisfy Condition 1; thus, an MRS cannot evacuate from the interior of the maze because it cannot find any exit.
Condition 2 ensures that an exit can be reached by trailing walls.
To evacuate from the maze in Fig.~\ref{fig:maze-conditions}(b), an MRS $R$ must place some kind of token on a cell to remember that $R$ has already searched the cell so that it does not search the same cell again.
However, the MRS (and modules) assumed here does not have such a function.
Thus, we assume that Condition 2 is satisfied.
\begin{figure}[t]
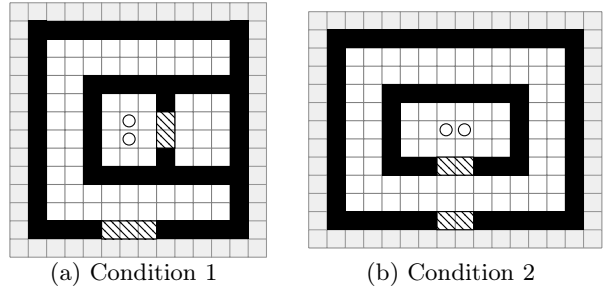

	\centering
	\begin{tabular}{cc}

\begin{customviewplot}{14}{14}
\outside{1}

\gwall{1}{1}{4}{1}
\gwall{8}{1}{12}{1}
\gwall{1}{1}{1}{12}
\gwall{12}{1}{12}{12}
\gwall{1}{12}{12}{12}
\gwall{4}{4}{12}{4}
\gwall{4}{9}{12}{9}
\gwall{4}{4}{4}{9}
\gwall{12}{4}{12}{9}
\gwall{8}{4}{8}{5}
\gwall{8}{8}{8}{9}
\gexit{5}{1}{7}{1}
\gexit{8}{6}{8}{7}

\gmodule{6}{6}
\gmodule{6}{7}

\end{customviewplot}
&
\begin{customviewplot}{16}{13}
\outside{1}

\gwall{1}{1}{6}{1}
\gwall{9}{1}{14}{1}
\gwall{1}{11}{14}{11}
\gwall{1}{1}{1}{11}
\gwall{14}{1}{14}{11}
\gwall{4}{4}{6}{4}
\gwall{9}{4}{11}{4}
\gwall{4}{4}{4}{8}
\gwall{11}{4}{11}{8}
\gwall{4}{8}{11}{8}
\gexit{7}{1}{8}{1}
\gexit{7}{4}{8}{4}

\gmodule{7}{6}
\gmodule{8}{6}

\end{customviewplot}
\\
\subcaption{(a) Condition 1}
&
\subcaption{(b) Condition 2}

\end{tabular}

	\caption{Counterexamples of the conditions that a maze must satisfy}
	\label{fig:maze-conditions}
\end{figure}

\section{Evacuation from a rectangular field}
\label{sec:evacuation-rectangular}
Here, we consider evacuation of an MRS from a rectangular field.
Section \ref{sec:evacuation-rectangular-with-global-compass} discusses the case where each module is equipped with a global compass, and Section \ref{sec:evacuation-rectangular-wo-global-compass} discusses the case where each module does not have a global compass.
Finally, Section \ref{sec:stop-after-evacuation} shows the possibility and impossibility of an evacuation algorithm to stop an MRS from moving after evacuation.

\subsection{Modules with a global compass}
\label{sec:evacuation-rectangular-with-global-compass}

Here, we prove the following theorem to show that $n=2$ is a necessary and sufficient condition for evacuation from any rectangular field given that the modules are equipped with a global compass.
\begin{theorem}
\label{thm:evacuation-with-gc}
Two modules equipped with a global compass are necessary and sufficient for a metamorphic robotic system to solve the evacuation problem in any rectangular field, starting from an arbitrary initial state.
\end{theorem}

The necessary condition of Theorem \ref{thm:evacuation-with-gc} (i.e., an MRS must be composed of at least two modules to be able to solve the evacuation problem) can be proven easily because any movement of a module requires at least one other backbone module.
In the remainder of this section, we prove that $n=2$ is sufficient by demonstrating an algorithm that solves the evacuation problem with two modules.

Movement of an MRS $R$ in the proposed algorithm is composed of the following two parts:
\begin{itemize}
	\item \textbf{Reaching a wall}: $R$ moves in one direction from its initial position and initial shape until reaching a wall.
	\item \textbf{Searching for an exit}: $R$ moves along the walls to find an exit and evacuates through that exit.
\end{itemize}
Figure \ref{fig:mrs-trail-in-rectangular-field} illustrates an example trail of an MRS performing these movements.
\begin{figure}[t]
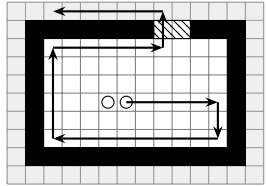

	\centering
	\begin{customviewplot}{14}{10}
\outside{1}
\gwall{1}{1}{1}{8}
\gwall{12}{1}{12}{8}
\gwall{1}{1}{12}{1}
\gwall{1}{8}{7}{8}
\gwall{10}{8}{12}{8}

\gexit{8}{8}{9}{8}

\gmodule{6}{4}
\gmodule{5}{4}

\draw[myarrow] (6.5,4.5) -- (11.5,4.5);
\draw[myarrow] (11.5,4.5) -- (11.5,2.5);
\draw[myarrow] (11.5,2.5) -- (2.5,2.5);
\draw[myarrow] (2.5,2.5) -- (2.5,7.5);
\draw[myarrow] (2.5,7.5) -- (8.5,7.5);
\draw[myarrow] (8.5,7.5) -- (8.5,9.5);
\draw[myarrow] (8.5,9.5) -- (2.5,9.5);
\end{customviewplot}
	\caption{An example evacuation path of an MRS from a rectangular field}
	\label{fig:mrs-trail-in-rectangular-field}
\end{figure}
The proposed algorithm uses the locomotion algorithm (Fig.~\ref{fig:locomotion-with-gc}) proposed in \cite{Chen2014} to reach a wall.
In the algorithm, module $a$ rotates in a clockwise direction around backbone module $b$ if the westmost module $a$ is side-adjacent to module $b$.
The MRS moves east by repeating these movements until a module of the MRS touches a wall cell.
\begin{figure}[t]
	\centering
	\input{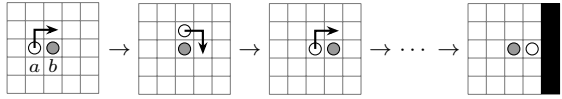}
	\caption{Locomotion by two modules with a global compass}
	\label{fig:locomotion-with-gc}
\end{figure}

The algorithm starts searching for an exit after reaching a wall.
The MRS moves along walls until reaching a corner or an exit, as depicted in Fig.~\ref{fig:move-along-wall-with-gc}.
While an MRS can move in an arbitrary direction (clockwise or counter-clockwise), we assume that an MRS moves so that the wall cells are always toward its left (i.e., the MRS moves clockwise).
\begin{figure}[t]
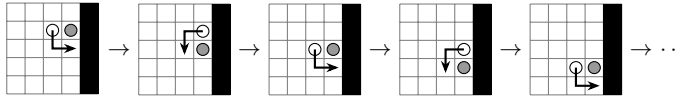

	\centering
	
\begin{tabular}{l}

\begin{viewplot}{5}
\wall{2}{3}{2}{-3}
\backbone{1}{1}
\module{0}{1}
\draw[myarrow] (2.5,3.5) -- (2.5,2.5) -- (3.8,2.5);
\end{viewplot}
\ra%
\begin{viewplot}{5}
\wall{2}{3}{2}{-3}
\module{1}{1}
\backbone{1}{0}
\draw[myarrow] (3.5,3.5) -- (2.5,3.5) -- (2.5,2.2);
\end{viewplot}
\ra%
\begin{viewplot}{5}
\wall{2}{3}{2}{-3}
\module{0}{0}
\backbone{1}{0}
\draw[myarrow] (2.5,2.5) -- (2.5,1.5) -- (3.8,1.5);
\end{viewplot}
\ra%
\begin{viewplot}{5}
\wall{2}{3}{2}{-3}
\backbone{1}{-1}
\module{1}{0}
\draw[myarrow] (3.5,2.5) -- (2.5,2.5) -- (2.5,1.2);
\end{viewplot}
\ra%
\begin{viewplot}{5}
\wall{2}{3}{2}{-3}
\backbone{1}{-1}
\module{0}{-1}
\draw[myarrow] (2.5,1.5) -- (2.5,0.5) -- (3.8,0.5);
\end{viewplot}
\ra\ $\cdots$

\end{tabular}

	\caption{Two modules with a global compass move along a wall}
	\label{fig:move-along-wall-with-gc}
\end{figure}
When an MRS arrives at a corner, it changes direction by rotating $\pi/2$ in a clockwise direction (Fig.~\ref{fig:turn-corner-with-gc}\footnote{Movements after changing direction may differ from those shown in Fig.~\ref{fig:turn-corner-with-gc} due to the different field sizes and exit locations. The figure shows only the case where the width of the interior is larger than three cells (Fig.~\ref{fig:turn-corner-with-gc}(a)) and the case where the width is two cells (Fig.~\ref{fig:turn-corner-with-gc}(b)). We omit the other cases here, but the movements of the modules are mostly the same as in Fig.~\ref{fig:turn-corner-with-gc}}).
After turning, the MRS moves along walls again with the movements shown in Fig.~\ref{fig:move-along-wall-with-gc}.
\begin{figure}[t]
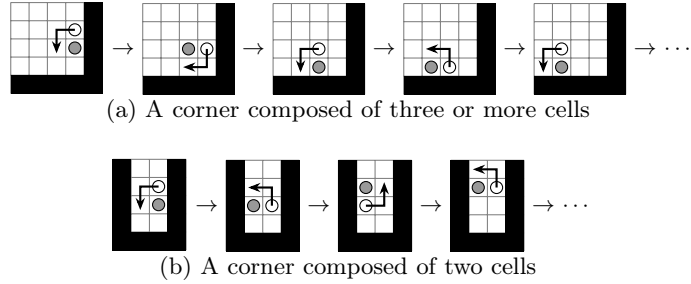

	\centering
	
\begin{tabular}{c}

\begin{tabular}{l}

\begin{viewplot}{5}
\wall{2}{3}{2}{-3}
\wall{-2}{-2}{2}{-2}
\module{1}{1}
\backbone{1}{0}
\draw[myarrow] (3.5,3.5) -- (2.5,3.5) -- (2.5,2.2);
\end{viewplot}
\ra
\begin{viewplot}{5}
\wall{2}{3}{2}{-3}
\wall{-2}{-2}{2}{-2}
\backbone{0}{0}
\module{1}{0}
\draw[myarrow] (3.5,2.5) -- (3.5,1.5) -- (2.2,1.5);
\end{viewplot}
\ra
\begin{viewplot}{5}
\wall{2}{3}{2}{-3}
\wall{-2}{-2}{2}{-2}
\backbone{0}{-1}
\module{0}{0}
\draw[myarrow] (2.5,2.5) -- (1.5,2.5) -- (1.5,1.2);
\end{viewplot}
\ra
\begin{viewplot}{5}
\wall{2}{3}{2}{-3}
\wall{-2}{-2}{2}{-2}
\backbone{-1}{-1}
\module{0}{-1}
\draw[myarrow] (2.5,1.5) -- (2.5,2.5) -- (1.2,2.5);
\end{viewplot}
\ra
\begin{viewplot}{5}
\wall{2}{3}{2}{-3}
\wall{-2}{-2}{2}{-2}
\backbone{-1}{-1}
\module{-1}{0}
\draw[myarrow] (1.5,2.5) -- (0.5,2.5) -- (0.5,1.2);
\end{viewplot}
\ra\ $\cdots$
\end{tabular}
\\
\subcaption{(a) A corner composed of three or more cells}\\
\\
\begin{customviewplot}{4}{5}
\wall{-2}{3}{-2}{-3}
\wall{1}{3}{1}{-3}
\wall{-2}{-2}{1}{-2}
\module{0}{1}
\backbone{0}{0}
\draw[myarrow] (2.5,3.5) -- (1.5,3.5) -- (1.5,2.2);
\end{customviewplot}
\ra
\begin{customviewplot}{4}{5}
\wall{-2}{3}{-2}{-3}
\wall{1}{3}{1}{-3}
\wall{-2}{-2}{1}{-2}
\backbone{-1}{0}
\module{0}{0}
\draw[myarrow] (2.5,2.5) -- (2.5,3.5) -- (1.2,3.5);
\end{customviewplot}
\ra
\begin{customviewplot}{4}{5}
\wall{-2}{3}{-2}{-3}
\wall{1}{3}{1}{-3}
\wall{-2}{-2}{1}{-2}
\module{-1}{0}
\backbone{-1}{1}
\draw[myarrow] (1.5,2.5) -- (2.5,2.5) -- (2.5,3.8);
\end{customviewplot}
\ra
\begin{customviewplot}{4}{5}
\wall{-2}{3}{-2}{-3}
\wall{1}{3}{1}{-3}
\wall{-2}{-2}{1}{-2}
\module{0}{1}
\backbone{-1}{1}
\draw[myarrow] (2.5,3.5) -- (2.5,4.5) -- (1.2,4.5);
\end{customviewplot}
\ra\ $\cdots$
\\
\subcaption{(b) A corner composed of two cells}

\end{tabular}

	\caption{Two modules with a global compass turn around a corner}
	\label{fig:turn-corner-with-gc}
\end{figure}

After reaching a wall, the MRS eventually finds an exit from the interior of a rectangular field by repeating the movements shown in Figs.~\ref{fig:move-along-wall-with-gc} and \ref{fig:turn-corner-with-gc}.
In a rectangular field, only exits have convex shapes; thus, an MRS can easily distinguish them from other wall cells.
The MRS evacuates from a rectangular field with the movements in Fig.~\ref{fig:evacuate-from-exit-with-gc} when it finds an exit.
\begin{figure}[t]
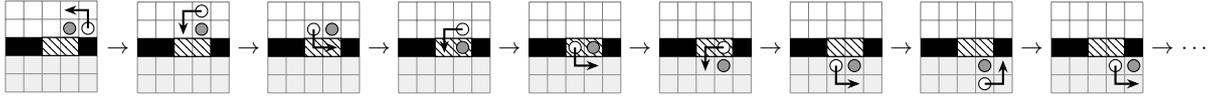

	\centering
	
\begin{tabular}{l}

\begin{viewplot}{5}
\wall{2}{0}{2}{0}
\wall{-2}{0}{-1}{0}
\outsidearea{0}{0}{4}{1}
\exit{0}{0}{1}{0}

\module{2}{1}
\backbone{1}{1}
\draw[myarrow] (4.5,3.5) -- (4.5,4.5) -- (3.2,4.5);
\end{viewplot}
\ra
\begin{viewplot}{5}
\wall{2}{0}{2}{0}
\wall{-2}{0}{-1}{0}
\outsidearea{0}{0}{4}{1}
\exit{0}{0}{1}{0}

\module{1}{2}
\backbone{1}{1}
\draw[myarrow] (3.5,4.5) -- (2.5,4.5) -- (2.5,3.2);
\end{viewplot}
\ra
\begin{viewplot}{5}
\wall{2}{0}{2}{0}
\wall{-2}{0}{-1}{0}
\outsidearea{0}{0}{4}{1}
\exit{0}{0}{1}{0}

\module{0}{1}
\backbone{1}{1}
\draw[myarrow] (2.5,3.5) -- (2.5,2.5) -- (3.8,2.5);
\end{viewplot}
\ra
\begin{viewplot}{5}
\wall{2}{0}{2}{0}
\wall{-2}{0}{-1}{0}
\outsidearea{0}{0}{4}{1}
\exit{0}{0}{1}{0}

\module{1}{1}
\backbone{1}{0}
\draw[myarrow] (3.5,3.5) -- (2.5,3.5) -- (2.5,2.2);
\end{viewplot}
\ra
\begin{viewplot}{5}
\wall{2}{0}{2}{0}
\wall{-2}{0}{-1}{0}
\outsidearea{0}{0}{4}{1}
\exit{0}{0}{1}{0}

\module{0}{0}
\backbone{1}{0}
\draw[myarrow] (2.5,2.5) -- (2.5,1.5) -- (3.8,1.5);
\end{viewplot}
\ra
\begin{viewplot}{5}
\wall{2}{0}{2}{0}
\wall{-2}{0}{-1}{0}
\outsidearea{0}{0}{4}{1}
\exit{0}{0}{1}{0}

\module{1}{0}
\backbone{1}{-1}
\draw[myarrow] (3.5,2.5) -- (2.5,2.5) -- (2.5,1.2);
\end{viewplot}
\ra
\begin{viewplot}{5}
\wall{2}{0}{2}{0}
\wall{-2}{0}{-1}{0}
\outsidearea{0}{0}{4}{1}
\exit{0}{0}{1}{0}

\module{0}{-1}
\backbone{1}{-1}
\draw[myarrow] (2.5,1.5) -- (2.5,0.5) -- (3.8,0.5);
\end{viewplot}
\ra
\begin{viewplot}{5}
\wall{2}{0}{2}{0}
\wall{-2}{0}{-1}{0}
\outsidearea{0}{0}{4}{1}
\exit{0}{0}{1}{0}

\module{1}{-2}
\backbone{1}{-1}
\draw[myarrow] (3.5,0.5) -- (4.5,0.5) -- (4.5,1.8);
\end{viewplot}
\ra
\begin{viewplot}{5}
\wall{2}{0}{2}{0}
\wall{-2}{0}{-1}{0}
\outsidearea{0}{0}{4}{1}
\exit{0}{0}{1}{0}

\backbone{2}{-1}
\module{1}{-1}
\draw[myarrow] (3.5,1.5) -- (3.5,0.5) -- (4.8,0.5);
\end{viewplot}
\ra\ $\cdots$

\end{tabular}

	\caption{Evacuation through an exit by two modules with a global compass}
	\label{fig:evacuate-from-exit-with-gc}
\end{figure}

We have proven that an MRS composed of two modules equipped with a global compass can evacuate from an arbitrary rectangular field.
Therefore, we have proven Theorem \ref{thm:evacuation-with-gc}.
A module in this algorithm requires a visibility range of $5 \times 5$ (i.e., 2-neighborhood).
This visibility range derives from the fact that a module decides its movement based on whether there is a wall cell side-adjacent to its neighbor module.
For instance, if the cell is not a wall cell (e.g., the first state in Fig.~\ref{fig:locomotion-with-gc}), a module rotates $\pi/2$ around its neighbor module in a clockwise direction.
Otherwise, the module must rotate in a counter-clockwise direction (e.g., the first state in Fig.~\ref{fig:move-along-wall-with-gc}).

\subsection{Modules without a global compass}
\label{sec:evacuation-rectangular-wo-global-compass}

In this section, we consider evacuation by modules not equipped with a global compass.
First, we consider evacuation starting from restricted initial states in Section \ref{sec:evacuation-without-gc-restricted-init-states}.
Then, we discuss how we can remove this restriction on initial shapes in Section \ref{sec:evacuation-without-gc-arbitrary-init-states}.
Finally, we show that an MRS composed of only two modules that are not equipped with a global compass can solve the evacuation problem by restricting the MRS's initial position to a cell that is side-adjacent to a wall.

\subsubsection{Starting from restricted initial states}
\label{sec:evacuation-without-gc-restricted-init-states}

We prove the following theorem to show that $n=4$ is a necessary and sufficient condition for evacuation from any rectangular field if the initial shape of an MRS is restricted.
\begin{theorem}
\label{thm:evacuation-without-gc}
Four modules not equipped with a global compass are necessary and sufficient for the metamorphic robotic system to solve the evacuation problem in any rectangular field starting from the allowed initial states.
\end{theorem}

We prove the necessary part of Theorem \ref{thm:evacuation-without-gc} with the following lemma:
\begin{lemma}
\label{lem:evacuation-without-gc-impossibility}
Consider the metamorphic robotic system $R$ composed of less than four modules that are not equipped with a global compass in a sufficiently large rectangular field.
For any deterministic algorithm $\mathcal{A}$, there exist initial states in which $R$ cannot evacuate from a rectangular field.
\end{lemma}
\begin{proof}
The proof mostly follows Lemma 5 in \cite{Doi2021}.
Here, we assume that each module cannot observe any wall in its initial state because walls are too far from its initial location.

For $n=2$, as depicted in Fig.~\ref{fig:initial-state-of-two-modules}, there is the case where each module tries to rotate around the other module.
This is impossible because there is no backbone module.
\begin{figure}[t]
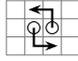

	\centering
	\begin{customviewplot}{4}{3}
	\gmodule{1}{1}
	\gmodule{2}{1}
	
	\draw[myarrow] (1.5,1.5) -- (1.5,0.5) -- (2.8,0.5);
	\draw[myarrow] (2.5,1.5) -- (2.5,2.5) -- (1.2,2.5);
\end{customviewplot}
	\caption{Impossible rotations by two modules without any backbone in an initial state}
	\label{fig:initial-state-of-two-modules}
\end{figure}

For $n=3$, there are initial states from which an MRS cannot move.
We classify these initial states into ``line-type'' and ``L-type'' and briefly describe each.

In the line-type initial state (Fig.~\ref{fig:evacuation-wo-gc-impossiblity-line-state}), the modules of both ends can move only by rotation.
If we assume that these modules rotate in the same direction (e.g., clockwise) to avoid collisions, the subsequent state is also line-type.
While the direction of the subsequent state is different from that of the first one, the only possible movement of the modules is a rotation.
This movement results in the initial line-type state.
\begin{figure}[t]
	\centering
	\begin{viewplot}{5}
\module{1}{0}
\backbone{0}{0}
\module{-1}{0}
\draw[myarrow] (1.5,2.5) -- (1.5,3.5) -- (2.8,3.5);
\draw[myarrow] (3.5,2.5) -- (3.5,1.5) -- (2.2,1.5);
\end{viewplot}
\ra
\begin{viewplot}{5}
\module{0}{1}
\backbone{0}{0}
\module{0}{-1}
\draw[myarrow] (2.5,3.5) -- (3.5,3.5) -- (3.5,2.2);
\draw[myarrow] (2.5,1.5) -- (1.5,1.5) -- (1.5,2.8);
\end{viewplot}
\ra
\begin{viewplot}{5}
\module{1}{0}
\backbone{0}{0}
\module{-1}{0}
\draw[myarrow] (1.5,2.5) -- (1.5,3.5) -- (2.8,3.5);
\draw[myarrow] (3.5,2.5) -- (3.5,1.5) -- (2.2,1.5);
\end{viewplot}
\ra\ $\cdots$
	\caption{Line-type forbidden initial states of three modules without a global compass}
	\label{fig:evacuation-wo-gc-impossiblity-line-state}
\end{figure}

We consider the L-type initial states illustrated in Fig.~\ref{fig:evacuation-wo-gc-impossiblity-l-state}.
In these initial states, only the left and the bottom modules can move.
However, if only one of them rotates around the middle module, the resulting shape is line-type (Figs.~\ref{fig:evacuation-wo-gc-impossiblity-l-state}(a) and (b)); thus, an MRS cannot move as we saw before.
If both modules rotate, the resulting state is L-type (Fig.~\ref{fig:evacuation-wo-gc-impossiblity-l-state}(c)).
In case of sliding, repeating four 1-slidings ends in the initial state, as depicted in Figs.~\ref{fig:evacuation-wo-gc-impossiblity-l-state}(d) and (e).
\begin{figure}[t]
	\centering
	\begin{tabular}{c}

	\begin{tabular}{ccc}

		\begin{viewplot}{5}
			\backbone{0}{0}
			\module{-1}{0}
			\backbone{0}{-1}
			\draw[myarrow] (1.5,2.5) -- (1.5,3.5) -- (2.8,3.5);
		\end{viewplot}
		\ra
		\begin{viewplot}{5}
			\backbone{0}{0}
			\backbone{0}{1}
			\backbone{0}{-1}
		\end{viewplot}
		&

		\begin{viewplot}{5}
			\backbone{0}{0}
			\backbone{-1}{0}
			\module{0}{-1}
			\draw[myarrow] (2.5,1.5) -- (3.5,1.5) -- (3.5,2.8);
		\end{viewplot}
		\ra
		\begin{viewplot}{5}
			\backbone{0}{0}
			\backbone{1}{0}
			\backbone{-1}{0}
		\end{viewplot}

		&

		\begin{viewplot}{5}
			\backbone{0}{0}
			\module{-1}{0}
			\module{0}{-1}
			\draw[myarrow] (1.5,2.5) -- (1.5,3.5) -- (2.8,3.5);
			\draw[myarrow] (2.5,1.5) -- (3.5,1.5) -- (3.5,2.8);
		\end{viewplot}
		\ra
		\begin{viewplot}{5}
			\backbone{0}{0}
			\backbone{1}{0}
			\backbone{0}{1}
		\end{viewplot}

		\\
		\subcaption{(a) Rotation 1} & \subcaption{(b) Rotation 2} & \subcaption{(c) Rotation 3} \\

	\end{tabular}
	\\ \\

	\begin{tabular}{c}

		\begin{viewplot}{4}
			\backbone{0}{0}
			\module{-1}{0}
			\backbone{0}{-1}
			\draw[myarrow] (1.5,2.5) -- (1.5,1.2);
		\end{viewplot}
		\ra
		\begin{viewplot}{4}
			\module{0}{0}
			\backbone{-1}{-1}
			\backbone{0}{-1}
			\draw[myarrow] (2.5,2.5) -- (1.2,2.5);
		\end{viewplot}
		\ra
		\begin{viewplot}{4}
			\backbone{-1}{-1}
			\backbone{-1}{0}
			\module{0}{-1}
			\draw[myarrow] (2.5,1.5) -- (2.5,2.8);
		\end{viewplot}
		\ra
		\begin{viewplot}{4}
			\module{-1}{-1}
			\backbone{-1}{0}
			\backbone{0}{0}
			\draw[myarrow] (1.5,1.5) -- (2.8,1.5);
		\end{viewplot}
		\ra
		\begin{viewplot}{4}
			\backbone{0}{0}
			\module{-1}{0}
			\backbone{0}{-1}
			\draw[myarrow] (1.5,2.5) -- (1.5,1.2);
		\end{viewplot}
		\ra\ $\cdots$
		\\
		\subcaption{(d) Sliding 1}
		\\ \\

		\begin{viewplot}{4}
			\backbone{0}{0}
			\backbone{-1}{0}
			\module{0}{-1}
			\draw[myarrow] (2.5,1.5) -- (1.2,1.5);
		\end{viewplot}
		\ra
		\begin{viewplot}{4}
			\module{0}{0}
			\backbone{-1}{-1}
			\backbone{-1}{0}
			\draw[myarrow] (2.5,2.5) -- (2.5,1.2);
		\end{viewplot}
		\ra
		\begin{viewplot}{4}
			\backbone{-1}{-1}
			\module{-1}{0}
			\backbone{0}{-1}
			\draw[myarrow] (1.5,2.5) -- (2.8,2.5);
		\end{viewplot}
		\ra
		\begin{viewplot}{4}
			\module{-1}{-1}
			\backbone{0}{-1}
			\backbone{0}{0}
			\draw[myarrow] (1.5,1.5) -- (1.5,2.8);
		\end{viewplot}
		\ra
		\begin{viewplot}{4}
			\backbone{0}{0}
			\backbone{-1}{0}
			\module{0}{-1}
			\draw[myarrow] (2.5,1.5) -- (1.2,1.5);
		\end{viewplot}
		\ra\ $\cdots$

		\\
		\subcaption{(e) Sliding 2} \\
	\end{tabular}

\end{tabular}

	\caption{L-type forbidden initial states of three modules without a global compass}
	\label{fig:evacuation-wo-gc-impossiblity-l-state}
\end{figure}

As observed above, there exist initial states for $n < 4$ where an MRS cannot move.
Because an MRS cannot find an exit from these initial states, it is impossible for it to evacuate from a rectangular field.
\end{proof}

Then, we show an evacuation algorithm for an MRS composed of four modules not equipped with a global compass.
This algorithm is not self-stabilizing because there are the initial shapes from which it cannot solve the evacuation problem.
\begin{lemma}
\label{lem:evacuation-without-gc-possibility}
Four modules not equipped with a global compass are sufficient for a metamorphic robotic system to solve the evacuation problem in any given rectangular field starting from any initial state from which evacuation is possible.
\end{lemma}
\noindent As shown in the proof of Lemma 5 in \cite{Doi2021}, for an MRS composed of four modules not equipped with a global compass, there are initial states from which an MRS cannot break the symmetry (Fig.~\ref{fig:evacuation-without-compass-forbidden-states}).
Therefore, we exclude these initial states from the proposed algorithm and consider only the remaining initial states (Fig.~\ref{fig:evacuation-without-compass-allowed-states}).
\begin{figure}[t]
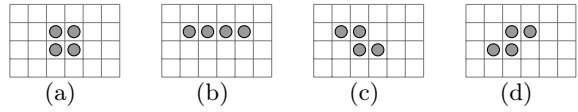

	\centering
	\begin{tabular}{cccc}
\begin{customviewplot}{6}{4}
\backbone{0}{0}
\backbone{-1}{0}
\backbone{0}{-1}
\backbone{-1}{-1}
\end{customviewplot}
&
\begin{customviewplot}{6}{4}
\backbone{0}{0}
\backbone{1}{0}
\backbone{-1}{0}
\backbone{-2}{0}
\end{customviewplot}
&
\begin{customviewplot}{6}{4}
\backbone{-1}{0}
\backbone{-2}{0}
\backbone{-1}{-1}
\backbone{0}{-1}
\end{customviewplot}
&
\begin{customviewplot}{6}{4}
\backbone{-1}{0}
\backbone{-2}{-1}
\backbone{-1}{-1}
\backbone{0}{0}
\end{customviewplot}
\\
\subcaption{(a)} & \subcaption{(b)} & \subcaption{(c)} & \subcaption{(d)} \\

\end{tabular}

	\caption{Forbidden initial states of four modules without a global compass}
	\label{fig:evacuation-without-compass-forbidden-states}
\end{figure}
\begin{figure}[t]
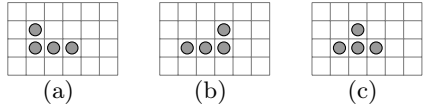

	\centering
	\begin{tabular}{ccc}
\begin{customviewplot}{6}{4}
\gbackbone{1}{2}
\gbackbone{1}{1}
\gbackbone{2}{1}
\gbackbone{3}{1}
\end{customviewplot}
&
\begin{customviewplot}{6}{4}
\gbackbone{3}{2}
\gbackbone{1}{1}
\gbackbone{2}{1}
\gbackbone{3}{1}
\end{customviewplot}
&
\begin{customviewplot}{6}{4}
\gbackbone{2}{2}
\gbackbone{1}{1}
\gbackbone{2}{1}
\gbackbone{3}{1}
\end{customviewplot}
\\
\subcaption{(a)} & \subcaption{(b)} & \subcaption{(c)} \\

\end{tabular}

	\caption{Allowed initial states of four modules without a global compass}
	\label{fig:evacuation-without-compass-allowed-states}
\end{figure}

As in Section \ref{sec:evacuation-rectangular-with-global-compass}, we build the proposed algorithm with two parts: reaching a wall and searching for an exit.
For reaching a wall, we use the locomotion algorithm for five modules proposed in \cite{Doi2021} with a slight modification for four modules (Fig.~\ref{fig:evacuation-wo-gc-one-way-move}).
\begin{figure}[t]
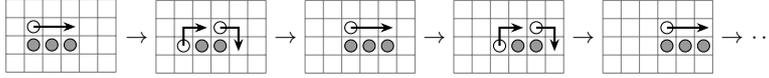

	\centering
	\begin{tabular}{l}

\begin{customviewplot}{6}{4}
\module{-2}{0}
\backbone{0}{-1}
\backbone{-1}{-1}
\backbone{-2}{-1}
\draw[myarrow] (1.5,2.5) -- (3.8,2.5);
\end{customviewplot}
\ra
\begin{customviewplot}{6}{4}
\module{0}{0}
\backbone{0}{-1}
\backbone{-1}{-1}
\module{-2}{-1}
\draw[myarrow] (1.5,1.5) -- (1.5,2.5) -- (2.8,2.5);
\draw[myarrow] (3.5,2.5) -- (4.5,2.5) -- (4.5,1.2);
\end{customviewplot}
\ra
\begin{customviewplot}{6}{4}
\module{-1}{0}
\backbone{1}{-1}
\backbone{0}{-1}
\backbone{-1}{-1}
\draw[myarrow] (2.5,2.5) -- (4.8,2.5);
\end{customviewplot}
\ra
\begin{customviewplot}{6}{4}
\module{1}{0}
\backbone{1}{-1}
\backbone{0}{-1}
\module{-1}{-1}
\draw[myarrow] (2.5,1.5) -- (2.5,2.5) -- (3.8,2.5);
\draw[myarrow] (4.5,2.5) -- (5.5,2.5) -- (5.5,1.2);
\end{customviewplot}
\ra
\begin{customviewplot}{6}{4}
\backbone{2}{-1}
\backbone{1}{-1}
\backbone{0}{-1}
\module{0}{0}
\draw[myarrow] (3.5,2.5) -- (5.8,2.5);
\end{customviewplot}
\ra\ $\cdots$

\end{tabular}

	\caption{Locomotion by four modules without a global compass}
	\label{fig:evacuation-wo-gc-one-way-move}
\end{figure}
After an MRS reaches a wall, the algorithm switches to searching for an exit.
The MRS changes its direction of movement as shown in Fig.~\ref{fig:evacuation-without-gc-move-along-wall} and moves along walls.
\begin{figure}[t]
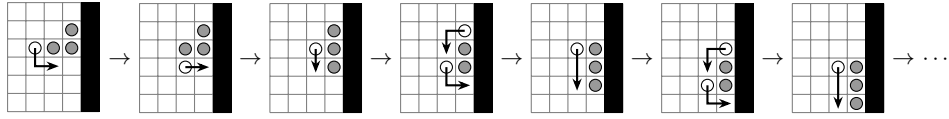

	\centering
	\begin{tabular}{l}

\begin{customviewplot}{5}{6}
\wall{2}{-3}{2}{2}
\backbone{0}{0}
\backbone{1}{0}
\backbone{1}{1}
\module{-1}{0}
\draw[myarrow] (1.5,3.5) -- (1.5,2.5) -- (2.8,2.5);
\end{customviewplot}
\ra
\begin{customviewplot}{5}{6}
\wall{2}{-3}{2}{2}
\backbone{0}{0}
\backbone{1}{1}
\backbone{1}{0}
\module{0}{-1}
\draw[myarrow] (2.5,2.5) -- (3.8,2.5);
\end{customviewplot}
\ra
\begin{customviewplot}{5}{6}
\wall{2}{-3}{2}{2}
\backbone{1}{-1}
\backbone{1}{1}
\backbone{1}{0}
\module{0}{0}
\draw[myarrow] (2.5,3.5) -- (2.5,2.2);
\end{customviewplot}
\ra
\begin{customviewplot}{5}{6}
\wall{2}{-3}{2}{2}
\backbone{1}{-1}
\backbone{1}{0}
\module{1}{1}
\module{0}{-1}
\draw[myarrow] (2.5,2.5) -- (2.5,1.5) -- (3.8, 1.5);
\draw[myarrow] (3.5,4.5) -- (2.5,4.5) -- (2.5, 3.2);
\end{customviewplot}
\ra
\begin{customviewplot}{5}{6}
\wall{2}{-3}{2}{2}
\backbone{1}{-1}
\backbone{1}{0}
\backbone{1}{-2}
\module{0}{0}
\draw[myarrow] (2.5,3.5) -- (2.5,1.2);
\end{customviewplot}
\ra
\begin{customviewplot}{5}{6}
\wall{2}{-3}{2}{2}
\backbone{1}{-1}
\backbone{1}{-2}
\module{0}{-2}
\module{1}{0}
\draw[myarrow] (2.5,1.5) -- (2.5,0.5) -- (3.8, 0.5);
\draw[myarrow] (3.5,3.5) -- (2.5,3.5) -- (2.5, 2.2);
\end{customviewplot}
\ra
\begin{customviewplot}{5}{6}
\wall{2}{-3}{2}{2}
\backbone{1}{-2}
\backbone{1}{-1}
\backbone{1}{-3}
\module{0}{-1}
\draw[myarrow] (2.5,2.5) -- (2.5,0.2);
\end{customviewplot}
\ra\ $\cdots$
\end{tabular}

	\caption{Four modules without a global compass move along walls}
	\label{fig:evacuation-without-gc-move-along-wall}
\end{figure}
When the MRS reaches a corner, it changes its direction by $\pi/2$ clockwise (Fig.~\ref{fig:evacuation-without-gc-turn-corner}).
\begin{figure}[t]
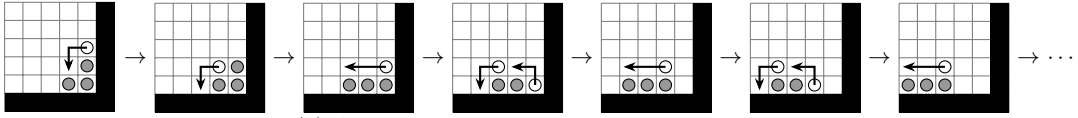
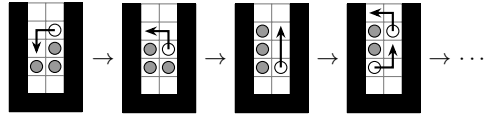

	\centering
	\begin{tabular}{c}

\begin{tabular}{l}

\begin{viewplot}{6}
\wall{2}{-3}{2}{2}
\wall{-3}{-3}{2}{-3}
\backbone{1}{-1}
\backbone{0}{-2}
\backbone{1}{-2}
\module{1}{0}
\draw[myarrow] (4.5,3.5) -- (3.5,3.5) -- (3.5,2.2);
\end{viewplot}
\ra
\begin{viewplot}{6}
\wall{2}{-3}{2}{2}
\wall{-3}{-3}{2}{-3}
\backbone{1}{-1}
\backbone{0}{-2}
\backbone{1}{-2}
\module{0}{-1}
\draw[myarrow] (3.5,2.5) -- (2.5,2.5) -- (2.5,1.2);
\end{viewplot}
\ra
\begin{viewplot}{6}
\wall{2}{-3}{2}{2}
\wall{-3}{-3}{2}{-3}
\backbone{-1}{-2}
\backbone{0}{-2}
\backbone{1}{-2}
\module{1}{-1}
\draw[myarrow] (4.5,2.5) -- (2.2,2.5);
\end{viewplot}
\ra
\begin{viewplot}{6}
\wall{2}{-3}{2}{2}
\wall{-3}{-3}{2}{-3}
\backbone{-1}{-2}
\backbone{0}{-2}
\module{-1}{-1}
\module{1}{-2}
\draw[myarrow] (2.5,2.5) -- (1.5,2.5) -- (1.5,1.2);
\draw[myarrow] (4.5,1.5) -- (4.5,2.5) -- (3.2,2.5);
\end{viewplot}
\ra
\begin{viewplot}{6}
\wall{2}{-3}{2}{2}
\wall{-3}{-3}{2}{-3}
\backbone{-1}{-2}
\backbone{-2}{-2}
\backbone{0}{-2}
\module{0}{-1}
\draw[myarrow] (3.5,2.5) -- (1.2,2.5);
\end{viewplot}
\ra
\begin{viewplot}{6}
\wall{2}{-3}{2}{2}
\wall{-3}{-3}{2}{-3}
\backbone{-2}{-2}
\backbone{-1}{-2}
\module{-2}{-1}
\module{0}{-2}
\draw[myarrow] (1.5,2.5) -- (0.5,2.5) -- (0.5,1.2);
\draw[myarrow] (3.5,1.5) -- (3.5,2.5) -- (2.2,2.5);
\end{viewplot}
\ra
\begin{viewplot}{6}
\wall{2}{-3}{2}{2}
\wall{-3}{-3}{2}{-3}
\backbone{-2}{-2}
\backbone{-3}{-2}
\backbone{-1}{-2}
\module{-1}{-1}
\draw[myarrow] (2.5,2.5) -- (0.2,2.5);
\end{viewplot}
\ra\ $\cdots$
\end{tabular}
\\

\subcaption{(a) A corner composed of three or more cells}\\

\\
\begin{tabular}{l}

\begin{customviewplot}{4}{6}
\gwall{0}{0}{0}{5}
\gwall{3}{0}{3}{5}
\gwall{1}{0}{2}{0}

\gbackbone{1}{2}
\gbackbone{2}{2}
\gbackbone{2}{3}
\gmodule{2}{4}

\draw[myarrow] (2.5,4.5) -- (1.5,4.5) -- (1.5,3.2);
\end{customviewplot}
\ra
\begin{customviewplot}{4}{6}
\gwall{0}{0}{0}{5}
\gwall{3}{0}{3}{5}
\gwall{1}{0}{2}{0}

\gbackbone{1}{2}
\gbackbone{2}{2}
\gbackbone{1}{3}
\gmodule{2}{3}

\draw[myarrow] (2.5,3.5) -- (2.5,4.5) -- (1.2,4.5);
\end{customviewplot}
\ra
\begin{customviewplot}{4}{6}
\gwall{0}{0}{0}{5}
\gwall{3}{0}{3}{5}
\gwall{1}{0}{2}{0}

\gbackbone{1}{2}
\gmodule{2}{2}
\gbackbone{1}{3}
\gbackbone{1}{4}

\draw[myarrow] (2.5,2.5) -- (2.5,4.8);
\end{customviewplot}
\ra
\begin{customviewplot}{4}{6}
\gwall{0}{0}{0}{5}
\gwall{3}{0}{3}{5}
\gwall{1}{0}{2}{0}

\gmodule{1}{2}
\gmodule{2}{4}
\gbackbone{1}{3}
\gbackbone{1}{4}

\draw[myarrow] (1.5,2.5) -- (2.5,2.5) -- (2.5,3.8);
\draw[myarrow] (2.5,4.5) -- (2.5,5.5) -- (1.2,5.5);
\end{customviewplot}
\ra\ $\cdots$
\end{tabular}
\\

\subcaption{(b) A corner composed of two cells}\\

\end{tabular}

	\caption{Four modules without a global compass turning a corner}
	\label{fig:evacuation-without-gc-turn-corner}
\end{figure}
When the MRS finds an exit, it evacuates from the rectangular field (Fig.~\ref{fig:evacuation-without-gc-exit-move}).
As we describe later, only in the case where we restrict initial shapes, we can stop an MRS from moving after evacuation.
\begin{figure}[t]
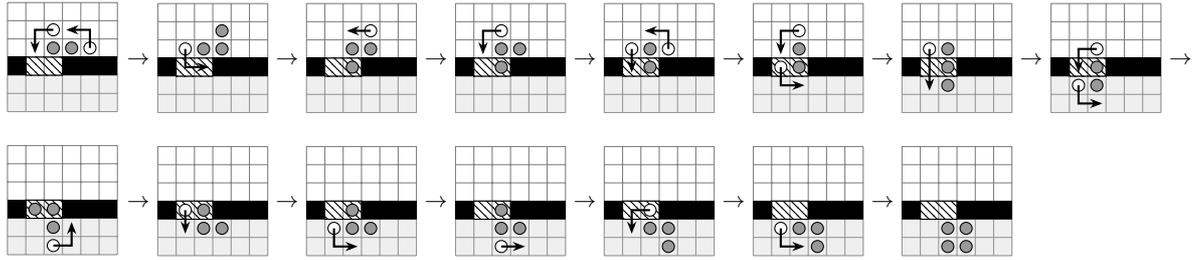

	\centering
	\begin{tabular}{l}

\begin{viewplot}{6}
\wall{-3}{-1}{-3}{-1}
\wall{0}{-1}{2}{-1}
\outsidearea{0}{0}{5}{1}
\gexit{1}{2}{2}{2}

\backbone{-1}{0}
\backbone{0}{0}
\module{-1}{1}
\module{1}{0}
\draw[myarrow] (2.5,4.5) -- (1.5,4.5) -- (1.5,3.2);
\draw[myarrow] (4.5,3.5) -- (4.5,4.5) -- (3.2,4.5);
\end{viewplot}
\ra
\begin{viewplot}{6}
\wall{-3}{-1}{-3}{-1}
\wall{0}{-1}{2}{-1}
\outsidearea{0}{0}{5}{1}
\gexit{1}{2}{2}{2}

\backbone{-1}{0}
\backbone{0}{0}
\backbone{0}{1}
\module{-2}{0}
\draw[myarrow] (1.5,3.5) -- (1.5,2.5) -- (2.8,2.5);
\end{viewplot}
\ra
\begin{viewplot}{6}
\wall{-3}{-1}{-3}{-1}
\wall{0}{-1}{2}{-1}
\outsidearea{0}{0}{5}{1}
\gexit{1}{2}{2}{2}

\backbone{-1}{0}
\backbone{0}{0}
\backbone{-1}{-1}
\module{0}{1}
\draw[myarrow] (3.5,4.5) -- (2.2,4.5);
\end{viewplot}
\ra
\begin{viewplot}{6}
\wall{-3}{-1}{-3}{-1}
\wall{0}{-1}{2}{-1}
\outsidearea{0}{0}{5}{1}
\gexit{1}{2}{2}{2}

\backbone{-1}{0}
\backbone{0}{0}
\backbone{-1}{-1}
\module{-1}{1}
\draw[myarrow] (2.5,4.5) -- (1.5,4.5) -- (1.5,3.2);
\end{viewplot}
\ra
\begin{viewplot}{6}
\wall{-3}{-1}{-3}{-1}
\wall{0}{-1}{2}{-1}
\outsidearea{0}{0}{5}{1}
\gexit{1}{2}{2}{2}

\backbone{-1}{0}
\backbone{-1}{-1}
\module{0}{0}
\module{-2}{0}
\draw[myarrow] (1.5,3.5) -- (1.5,2.2);
\draw[myarrow] (3.5,3.5) -- (3.5,4.5) -- (2.2,4.5);
\end{viewplot}
\ra
\begin{viewplot}{6}
\wall{-3}{-1}{-3}{-1}
\wall{0}{-1}{2}{-1}
\outsidearea{0}{0}{5}{1}
\gexit{1}{2}{2}{2}

\backbone{-1}{0}
\backbone{-1}{-1}
\module{-2}{-1}
\module{-1}{1}
\draw[myarrow] (1.5,2.5) -- (1.5,1.5) -- (2.8,1.5);
\draw[myarrow] (2.5,4.5) -- (1.5,4.5) -- (1.5,3.2);
\end{viewplot}
\ra
\begin{viewplot}{6}
\wall{-3}{-1}{-3}{-1}
\wall{0}{-1}{2}{-1}
\outsidearea{0}{0}{5}{1}
\gexit{1}{2}{2}{2}

\backbone{-1}{0}
\backbone{-1}{-2}
\backbone{-1}{-1}
\module{-2}{0}
\draw[myarrow] (1.5,3.5) -- (1.5,1.2);
\end{viewplot}
\ra
\begin{viewplot}{6}
\wall{-3}{-1}{-3}{-1}
\wall{0}{-1}{2}{-1}
\outsidearea{0}{0}{5}{1}
\gexit{1}{2}{2}{2}

\backbone{-1}{-2}
\backbone{-1}{-1}
\module{-1}{0}
\module{-2}{-2}
\draw[myarrow] (1.5,1.5) -- (1.5,0.5) -- (2.8,0.5);
\draw[myarrow] (2.5,3.5) -- (1.5,3.5) -- (1.5,2.2);
\end{viewplot}
\ra
\\ \\
\begin{viewplot}{6}
\wall{-3}{-1}{-3}{-1}
\wall{0}{-1}{2}{-1}
\outsidearea{0}{0}{5}{1}
\gexit{1}{2}{2}{2}

\backbone{-1}{-2}
\backbone{-1}{-1}
\backbone{-2}{-1}
\module{-1}{-3}
\draw[myarrow] (2.5,0.5) -- (3.5,0.5) -- (3.5,1.8);
\end{viewplot}
\ra
\begin{viewplot}{6}
\wall{-3}{-1}{-3}{-1}
\wall{0}{-1}{2}{-1}
\outsidearea{0}{0}{5}{1}
\gexit{1}{2}{2}{2}

\backbone{0}{-2}
\backbone{-1}{-2}
\backbone{-1}{-1}
\module{-2}{-1}
\draw[myarrow] (1.5,2.5) -- (1.5,1.2);
\end{viewplot}
\ra
\begin{viewplot}{6}
\wall{-3}{-1}{-3}{-1}
\wall{0}{-1}{2}{-1}
\outsidearea{0}{0}{5}{1}
\gexit{1}{2}{2}{2}

\backbone{0}{-2}
\backbone{-1}{-2}
\backbone{-1}{-1}
\module{-2}{-2}
\draw[myarrow] (1.5,1.5) -- (1.5,0.5) -- (2.8,0.5);
\end{viewplot}
\ra
\begin{viewplot}{6}
\wall{-3}{-1}{-3}{-1}
\wall{0}{-1}{2}{-1}
\outsidearea{0}{0}{5}{1}
\gexit{1}{2}{2}{2}

\backbone{0}{-2}
\backbone{-1}{-2}
\backbone{-1}{-1}
\module{-1}{-3}
\draw[myarrow] (2.5,0.5) -- (3.8,0.5);
\end{viewplot}
\ra
\begin{viewplot}{6}
\wall{-3}{-1}{-3}{-1}
\wall{0}{-1}{2}{-1}
\outsidearea{0}{0}{5}{1}
\gexit{1}{2}{2}{2}

\backbone{0}{-2}
\backbone{-1}{-2}
\module{-1}{-1}
\backbone{0}{-3}
\draw[myarrow] (2.5,2.5) -- (1.5,2.5) -- (1.5,1.2);
\end{viewplot}
\ra
\begin{viewplot}{6}
\wall{-3}{-1}{-3}{-1}
\wall{0}{-1}{2}{-1}
\outsidearea{0}{0}{5}{1}
\gexit{1}{2}{2}{2}

\backbone{0}{-2}
\backbone{-1}{-2}
\module{-2}{-2}
\backbone{0}{-3}
\draw[myarrow] (1.5,1.5) -- (1.5,0.5) -- (2.8,0.5);
\end{viewplot}
\ra
\begin{viewplot}{6}
\wall{-3}{-1}{-3}{-1}
\wall{0}{-1}{2}{-1}
\outsidearea{0}{0}{5}{1}
\gexit{1}{2}{2}{2}

\backbone{0}{-2}
\backbone{-1}{-2}
\backbone{-1}{-3}
\backbone{0}{-3}
\end{viewplot}

\end{tabular}

	\caption{Four modules without a global compass evacuating through an exit}
	\label{fig:evacuation-without-gc-exit-move}
\end{figure}

We have thus, proved that an MRS composed of four modules that are not equipped with a global compass can evacuate from any rectangular field if the initial shape of an MRS is restricted; thus, Theorem \ref{thm:evacuation-without-gc} was proven.
This algorithm requires a visibility range of $7 \times 7$ (i.e., 3-neighborhood) for the same reason described in Section \ref{sec:evacuation-rectangular-with-global-compass}.

\subsubsection{Starting from arbitrary initial states}
\label{sec:evacuation-without-gc-arbitrary-init-states}

The limitation on the initial states in Fig.~\ref{fig:evacuation-without-compass-forbidden-states} follows from the fact that an MRS cannot move to a wall from these initial states.
The locomotion algorithm for seven or more modules \cite{Doi2021} illustrated in Fig.~\ref{fig:locomotion-with-seven-modules} can resolve this limitation as stated in the following lemma:
\begin{figure}[t]
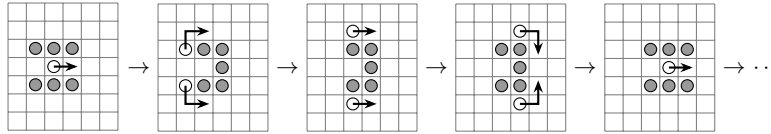

	\centering
	\begin{tabular}{l}

\begin{customviewplot}{6}{7}
	\gbackbone{3}{2}
	\gbackbone{1}{2}
	\gbackbone{2}{2}

	\gbackbone{3}{4}
	\gbackbone{1}{4}
	\gbackbone{2}{4}

	\gmodule{2}{3}

	\draw[myarrow] (2.5,3.5) -- (3.8,3.5);
\end{customviewplot}
\ra
\begin{customviewplot}{6}{7}
	\gmodule{1}{2}
	\gbackbone{2}{2}
	\gbackbone{3}{2}

	\gmodule{1}{4}
	\gbackbone{2}{4}
	\gbackbone{3}{4}

	\gbackbone{3}{3}

	\draw[myarrow] (1.5,2.5) -- (1.5,1.5) -- (2.8,1.5);
	\draw[myarrow] (1.5,4.5) -- (1.5,5.5) -- (2.8,5.5);
\end{customviewplot}
\ra
\begin{customviewplot}{6}{7}
	\gmodule{2}{1}
	\gbackbone{2}{2}
	\gbackbone{3}{2}

	\gmodule{2}{5}
	\gbackbone{2}{4}
	\gbackbone{3}{4}

	\gbackbone{3}{3}

	\draw[myarrow] (2.5,1.5) -- (3.8,1.5);
	\draw[myarrow] (2.5,5.5) -- (3.8,5.5);
\end{customviewplot}
\ra
\begin{customviewplot}{6}{7}
	\gmodule{3}{1}
	\gbackbone{2}{2}
	\gbackbone{3}{2}

	\gmodule{3}{5}
	\gbackbone{2}{4}
	\gbackbone{3}{4}

	\gbackbone{3}{3}

	\draw[myarrow] (3.5,1.5) -- (4.5,1.5) -- (4.5,2.8);
	\draw[myarrow] (3.5,5.5) -- (4.5,5.5) -- (4.5,4.2);
\end{customviewplot}
\ra
\begin{customviewplot}{6}{7}
	\gbackbone{4}{2}
	\gbackbone{2}{2}
	\gbackbone{3}{2}

	\gbackbone{4}{4}
	\gbackbone{2}{4}
	\gbackbone{3}{4}

	\gmodule{3}{3}

	\draw[myarrow] (3.5,3.5) -- (4.8,3.5);
\end{customviewplot}
\ra\ $\cdots$

\end{tabular}

	\caption{Locomotion by seven modules without a global compass}
	\label{fig:locomotion-with-seven-modules}
\end{figure}
\begin{lemma}[from \cite{Doi2021}]
Assuming a common handedness among the modules, seven modules not equipped with a global compass are necessary and sufficient for the metamorphic robotic system to perform locomotion from an arbitrary initial state.
\end{lemma}
Namely, an MRS reaches a wall by the locomotion algorithm.
Movement while searching for an exit is similar to that described in Section \ref{sec:evacuation-without-gc-restricted-init-states}.
An MRS changes its direction of movement and moves along walls (Fig.~\ref{fig:move-along-wall-wo-gc-seven}).
Then, it turns a corner (Fig.~\ref{fig:turn-corner-wo-gc-seven}).
Eventually, it finds an exit and evacuates from the rectangular field through that exit (Fig.~\ref{fig:evacuate-from-exit-wo-gc-seven}).
\begin{figure}[t]
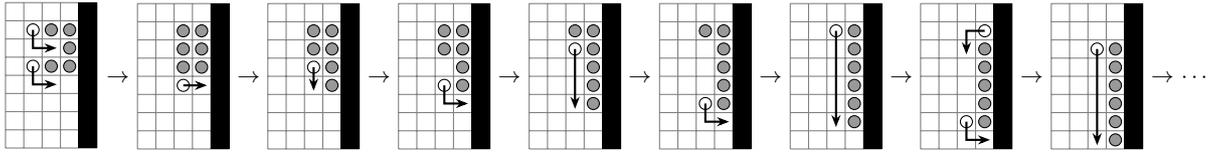

	\centering
	\begin{tabular}{l}

\begin{customviewplot}{5}{8}
	\gwall{4}{0}{4}{7}

	\gbackbone{3}{6}
	\gbackbone{3}{5}
	\gbackbone{3}{4}

	\gbackbone{2}{6}
	\gbackbone{2}{4}

	\gmodule{1}{6}
	\gmodule{1}{4}

	\draw[myarrow] (1.5,6.5) -- (1.5,5.5) -- (2.8,5.5);
	\draw[myarrow] (1.5,4.5) -- (1.5,3.5) -- (2.8,3.5);
\end{customviewplot}
\ra
\begin{customviewplot}{5}{8}
	\gwall{4}{0}{4}{7}

	\gbackbone{3}{6}
	\gbackbone{3}{5}
	\gbackbone{3}{4}

	\gbackbone{2}{6}
	\gbackbone{2}{4}

	\gbackbone{2}{5}
	\gmodule{2}{3}

	\draw[myarrow] (2.5,3.5) -- (3.8,3.5);
\end{customviewplot}
\ra
\begin{customviewplot}{5}{8}
	\gwall{4}{0}{4}{7}

	\gbackbone{3}{6}
	\gbackbone{3}{5}
	\gbackbone{3}{4}

	\gbackbone{2}{6}
	\gmodule{2}{4}

	\gbackbone{2}{5}
	\gbackbone{3}{3}

	\draw[myarrow] (2.5,4.5) -- (2.5,3.2);
\end{customviewplot}
\ra
\begin{customviewplot}{5}{8}
	\gwall{4}{0}{4}{7}

	\gbackbone{3}{6}
	\gbackbone{3}{5}
	\gbackbone{3}{4}

	\gbackbone{2}{6}
	\gmodule{2}{3}

	\gbackbone{2}{5}
	\gbackbone{3}{3}

	\draw[myarrow] (2.5,3.5) -- (2.5,2.5) -- (3.8,2.5);
\end{customviewplot}
\ra
\begin{customviewplot}{5}{8}
	\gwall{4}{0}{4}{7}

	\gbackbone{3}{6}
	\gbackbone{3}{5}
	\gbackbone{3}{4}

	\gbackbone{2}{6}
	\gbackbone{3}{2}

	\gmodule{2}{5}
	\gbackbone{3}{3}

	\draw[myarrow] (2.5,5.5) -- (2.5,2.2);
\end{customviewplot}
\ra
\begin{customviewplot}{5}{8}
	\gwall{4}{0}{4}{7}

	\gbackbone{3}{6}
	\gbackbone{3}{5}
	\gbackbone{3}{4}

	\gbackbone{2}{6}
	\gbackbone{3}{2}

	\gmodule{2}{2}
	\gbackbone{3}{3}

	\draw[myarrow] (2.5,2.5) -- (2.5,1.5) -- (3.8,1.5);
\end{customviewplot}
\ra
\begin{customviewplot}{5}{8}
	\gwall{4}{0}{4}{7}

	\gbackbone{3}{6}
	\gbackbone{3}{5}
	\gbackbone{3}{4}

	\gmodule{2}{6}
	\gbackbone{3}{2}

	\gbackbone{3}{1}
	\gbackbone{3}{3}

	\draw[myarrow] (2.5,6.5) -- (2.5,1.2);
\end{customviewplot}
\ra
\begin{customviewplot}{5}{8}
	\gwall{4}{0}{4}{7}

	\gmodule{3}{6}
	\gbackbone{3}{5}
	\gbackbone{3}{4}

	\gmodule{2}{1}
	\gbackbone{3}{2}

	\gbackbone{3}{1}
	\gbackbone{3}{3}

	\draw[myarrow] (2.5,1.5) -- (2.5,0.5) -- (3.8,0.5);
	\draw[myarrow] (3.5,6.5) -- (2.5,6.5) -- (2.5,5.2);
\end{customviewplot}
\ra
\begin{customviewplot}{5}{8}
	\gwall{4}{0}{4}{7}

	\gbackbone{3}{5}
	\gbackbone{3}{4}
	\gbackbone{3}{3}

	\gmodule{2}{5}
	\gbackbone{3}{1}

	\gbackbone{3}{0}
	\gbackbone{3}{2}

	\draw[myarrow] (2.5,5.5) -- (2.5,0.2);
\end{customviewplot}
\ra\ $\cdots$
\end{tabular}

	\caption{Seven modules without a global compass moving along walls}
	\label{fig:move-along-wall-wo-gc-seven}
\end{figure}
\begin{figure}[t]
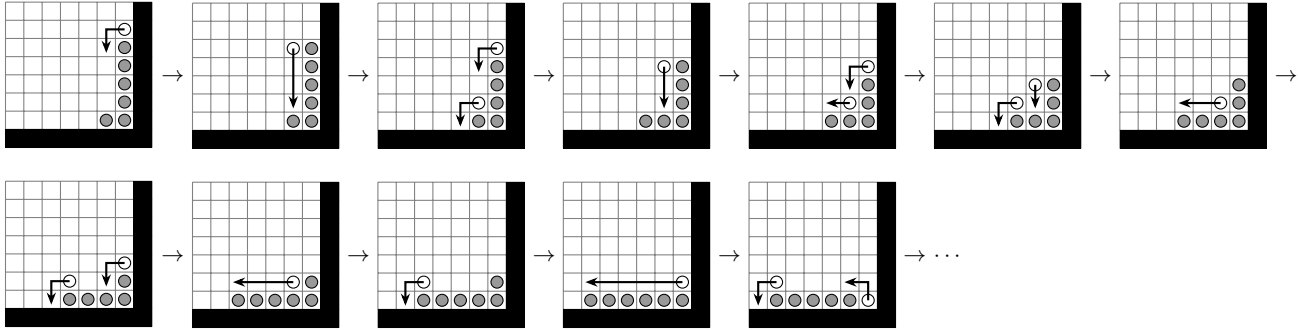

	\centering
	\begin{tabular}{l}

\begin{viewplot}{8}
	\gwall{0}{0}{7}{0}
	\gwall{7}{0}{7}{7}
	\gbackbone{5}{1}
	\gbackbone{6}{1}
	\gbackbone{6}{2}
	\gbackbone{6}{3}
	\gbackbone{6}{4}
	\gbackbone{6}{5}
	\gmodule{6}{6}
	\draw[myarrow] (6.5,6.5) -- (5.5,6.5) -- (5.5,5.2);
\end{viewplot}
\ra
\begin{viewplot}{8}
	\gwall{0}{0}{7}{0}
	\gwall{7}{0}{7}{7}
	\gbackbone{5}{1}
	\gbackbone{6}{1}
	\gbackbone{6}{2}
	\gbackbone{6}{3}
	\gbackbone{6}{4}
	\gbackbone{6}{5}
	\gmodule{5}{5}
	\draw[myarrow] (5.5,5.5) -- (5.5,2.2);
\end{viewplot}
\ra
\begin{viewplot}{8}
	\gwall{0}{0}{7}{0}
	\gwall{7}{0}{7}{7}
	\gbackbone{5}{1}
	\gbackbone{6}{1}
	\gbackbone{6}{2}
	\gbackbone{6}{3}
	\gbackbone{6}{4}
	\gmodule{6}{5}
	\gmodule{5}{2}
	\draw[myarrow] (5.5,2.5) -- (4.5,2.5) -- (4.5,1.2);
	\draw[myarrow] (6.5,5.5) -- (5.5,5.5) -- (5.5,4.2);
\end{viewplot}
\ra
\begin{viewplot}{8}
	\gwall{0}{0}{7}{0}
	\gwall{7}{0}{7}{7}
	\gbackbone{5}{1}
	\gbackbone{6}{1}
	\gbackbone{6}{2}
	\gbackbone{6}{3}
	\gbackbone{6}{4}
	\gbackbone{4}{1}
	\gmodule{5}{4}
	\draw[myarrow] (5.5,4.5) -- (5.5,2.2);
\end{viewplot}
\ra
\begin{viewplot}{8}
	\gwall{0}{0}{7}{0}
	\gwall{7}{0}{7}{7}
	\gbackbone{5}{1}
	\gbackbone{6}{1}
	\gbackbone{6}{2}
	\gbackbone{6}{3}
	\gbackbone{4}{1}
	\gmodule{6}{4}
	\gmodule{5}{2}
	\draw[myarrow] (6.5,4.5) -- (5.5,4.5) -- (5.5,3.2);
	\draw[myarrow] (5.5,2.5) -- (4.2,2.5);
\end{viewplot}
\ra
\begin{viewplot}{8}
	\gwall{0}{0}{7}{0}
	\gwall{7}{0}{7}{7}
	\gbackbone{5}{1}
	\gbackbone{6}{1}
	\gbackbone{6}{2}
	\gbackbone{6}{3}
	\gbackbone{4}{1}
	\gmodule{5}{3}
	\gmodule{4}{2}
	\draw[myarrow] (4.5,2.5) -- (3.5,2.5) -- (3.5,1.2);
	\draw[myarrow] (5.5,3.5) -- (5.5,2.2);
\end{viewplot}
\ra
\begin{viewplot}{8}
	\gwall{0}{0}{7}{0}
	\gwall{7}{0}{7}{7}
	\gbackbone{5}{1}
	\gbackbone{6}{1}
	\gbackbone{6}{2}
	\gbackbone{6}{3}
	\gbackbone{4}{1}
	\gbackbone{3}{1}
	\gmodule{5}{2}
	\draw[myarrow] (5.5,2.5) -- (3.2,2.5);
\end{viewplot}
\ra
\\ \\
\begin{viewplot}{8}
	\gwall{0}{0}{7}{0}
	\gwall{7}{0}{7}{7}
	\gbackbone{5}{1}
	\gbackbone{6}{1}
	\gbackbone{6}{2}
	\gbackbone{4}{1}
	\gbackbone{3}{1}
	\gmodule{6}{3}
	\gmodule{3}{2}
	\draw[myarrow] (3.5,2.5) -- (2.5,2.5) -- (2.5,1.2);
	\draw[myarrow] (6.5,3.5) -- (5.5,3.5) -- (5.5,2.2);
\end{viewplot}
\ra
\begin{viewplot}{8}
	\gwall{0}{0}{7}{0}
	\gwall{7}{0}{7}{7}
	\gbackbone{5}{1}
	\gbackbone{6}{1}
	\gbackbone{6}{2}
	\gbackbone{4}{1}
	\gbackbone{3}{1}
	\gbackbone{2}{1}
	\gmodule{5}{2}
	\draw[myarrow] (5.5,2.5) -- (2.2,2.5);
\end{viewplot}
\ra
\begin{viewplot}{8}
	\gwall{0}{0}{7}{0}
	\gwall{7}{0}{7}{7}
	\gbackbone{5}{1}
	\gbackbone{6}{1}
	\gbackbone{6}{2}
	\gbackbone{4}{1}
	\gbackbone{3}{1}
	\gbackbone{2}{1}
	\gmodule{2}{2}
	\draw[myarrow] (2.5,2.5) -- (1.5,2.5) -- (1.5,1.2);
\end{viewplot}
\ra
\begin{viewplot}{8}
	\gwall{0}{0}{7}{0}
	\gwall{7}{0}{7}{7}
	\gbackbone{6}{1}
	\gbackbone{5}{1}
	\gbackbone{4}{1}
	\gbackbone{3}{1}
	\gbackbone{2}{1}
	\gbackbone{1}{1}
	\gmodule{6}{2}
	\draw[myarrow] (6.5,2.5) -- (1.2,2.5);
\end{viewplot}
\ra
\begin{viewplot}{8}
	\gwall{0}{0}{7}{0}
	\gwall{7}{0}{7}{7}
	\gbackbone{5}{1}
	\gbackbone{4}{1}
	\gbackbone{3}{1}
	\gbackbone{2}{1}
	\gbackbone{1}{1}
	\gmodule{1}{2}
	\gmodule{6}{1}
	\draw[myarrow] (1.5,2.5) -- (0.5,2.5) -- (0.5,1.2);
	\draw[myarrow] (6.5,1.5) -- (6.5,2.5) -- (5.2,2.5);
\end{viewplot}
\ra\ $\cdots$
\end{tabular}

	\caption{Seven modules without a global compass turning a corner}
	\label{fig:turn-corner-wo-gc-seven}
\end{figure}
\begin{figure}[t]
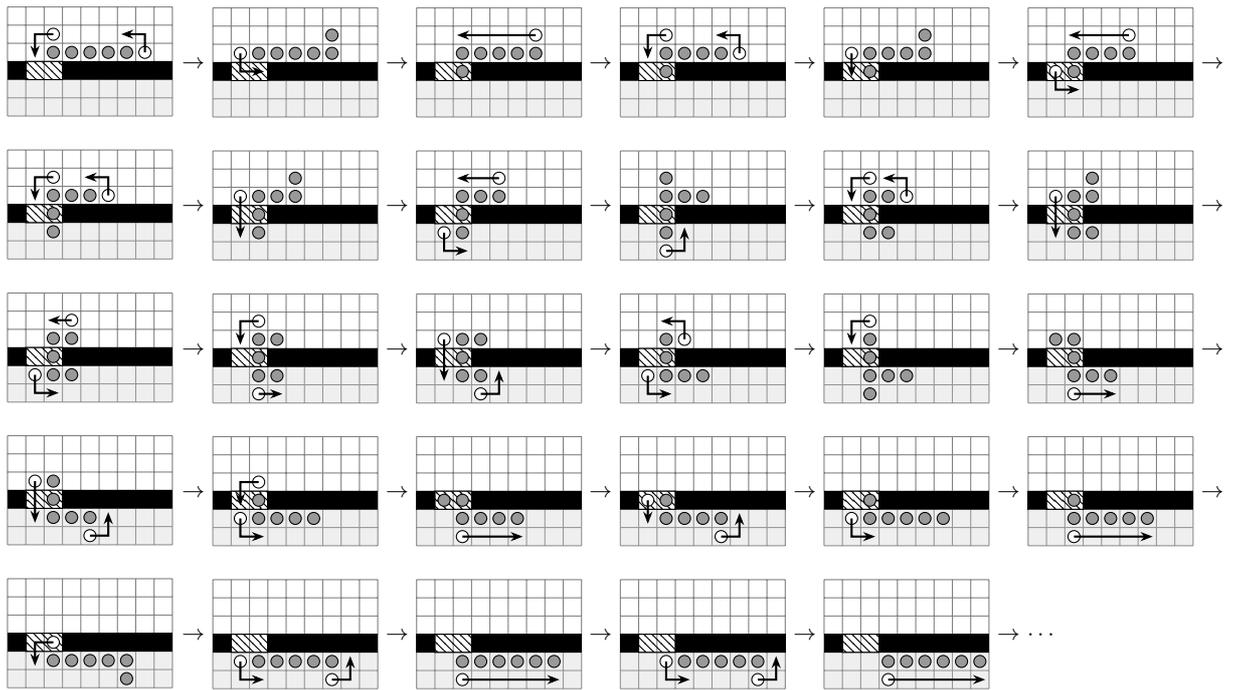

	\centering
	\begin{tabular}{l}

\begin{customviewplot}{9}{6}
	\gwall{0}{2}{0}{2}
	\gwall{3}{2}{8}{2}
	\outsidearea{0}{0}{8}{1}
	\gexit{1}{2}{2}{2}

	\gbackbone{2}{3}
	\gbackbone{3}{3}
	\gbackbone{4}{3}
	\gbackbone{5}{3}
	\gbackbone{6}{3}
	\gmodule{7}{3}
	\gmodule{2}{4}

	\draw[myarrow] (2.5,4.5) -- (1.5,4.5) -- (1.5,3.2);
	\draw[myarrow] (7.5,3.5) -- (7.5,4.5) -- (6.2,4.5);
\end{customviewplot}
\ra
\begin{customviewplot}{9}{6}
	\gwall{0}{2}{0}{2}
	\gwall{3}{2}{8}{2}
	\outsidearea{0}{0}{8}{1}
	\gexit{1}{2}{2}{2}

	\gbackbone{2}{3}
	\gbackbone{3}{3}
	\gbackbone{4}{3}
	\gbackbone{5}{3}
	\gbackbone{6}{3}
	\gbackbone{6}{4}
	\gmodule{1}{3}

	\draw[myarrow] (1.5,3.5) -- (1.5,2.5) -- (2.8,2.5);
\end{customviewplot}
\ra
\begin{customviewplot}{9}{6}
	\gwall{0}{2}{0}{2}
	\gwall{3}{2}{8}{2}
	\outsidearea{0}{0}{8}{1}
	\gexit{1}{2}{2}{2}

	\gbackbone{2}{3}
	\gbackbone{3}{3}
	\gbackbone{4}{3}
	\gbackbone{5}{3}
	\gbackbone{6}{3}
	\gmodule{6}{4}
	\gbackbone{2}{2}

	\draw[myarrow] (6.5,4.5) -- (2.2,4.5);
\end{customviewplot}
\ra
\begin{customviewplot}{9}{6}
	\gwall{0}{2}{0}{2}
	\gwall{3}{2}{8}{2}
	\outsidearea{0}{0}{8}{1}
	\gexit{1}{2}{2}{2}

	\gbackbone{2}{3}
	\gbackbone{3}{3}
	\gbackbone{4}{3}
	\gbackbone{5}{3}
	\gmodule{6}{3}
	\gmodule{2}{4}
	\gbackbone{2}{2}

	\draw[myarrow] (2.5,4.5) -- (1.5,4.5) -- (1.5,3.2);
	\draw[myarrow] (6.5,3.5) -- (6.5,4.5) -- (5.2,4.5);
\end{customviewplot}
\ra
\begin{customviewplot}{9}{6}
	\gwall{0}{2}{0}{2}
	\gwall{3}{2}{8}{2}
	\outsidearea{0}{0}{8}{1}
	\gexit{1}{2}{2}{2}

	\gbackbone{2}{3}
	\gbackbone{3}{3}
	\gbackbone{4}{3}
	\gbackbone{5}{3}
	\gbackbone{5}{4}
	\gmodule{1}{3}
	\gbackbone{2}{2}

	\draw[myarrow] (1.5,3.5) -- (1.5,2.2);
\end{customviewplot}
\ra
\begin{customviewplot}{9}{6}
	\gwall{0}{2}{0}{2}
	\gwall{3}{2}{8}{2}
	\outsidearea{0}{0}{8}{1}
	\gexit{1}{2}{2}{2}

	\gbackbone{2}{3}
	\gbackbone{3}{3}
	\gbackbone{4}{3}
	\gbackbone{5}{3}
	\gmodule{5}{4}
	\gmodule{1}{2}
	\gbackbone{2}{2}

	\draw[myarrow] (1.5,2.5) -- (1.5,1.5) -- (2.8,1.5);
	\draw[myarrow] (5.5,4.5) -- (2.2,4.5);
\end{customviewplot}
\ra
\\ \\
\begin{customviewplot}{9}{6}
	\gwall{0}{2}{0}{2}
	\gwall{3}{2}{8}{2}
	\outsidearea{0}{0}{8}{1}
	\gexit{1}{2}{2}{2}

	\gbackbone{2}{3}
	\gbackbone{3}{3}
	\gbackbone{4}{3}
	\gmodule{5}{3}
	\gbackbone{2}{1}
	\gbackbone{2}{2}
	\gmodule{2}{4}

	\draw[myarrow] (2.5,4.5) -- (1.5,4.5) -- (1.5,3.2);
	\draw[myarrow] (5.5,3.5) -- (5.5,4.5) -- (4.2,4.5);
\end{customviewplot}
\ra
\begin{customviewplot}{9}{6}
	\gwall{0}{2}{0}{2}
	\gwall{3}{2}{8}{2}
	\outsidearea{0}{0}{8}{1}
	\gexit{1}{2}{2}{2}

	\gbackbone{2}{3}
	\gbackbone{3}{3}
	\gbackbone{4}{3}
	\gbackbone{4}{4}
	\gbackbone{2}{1}
	\gbackbone{2}{2}
	\gmodule{1}{3}

	\draw[myarrow] (1.5,3.5) -- (1.5,1.2);
\end{customviewplot}
\ra
\begin{customviewplot}{9}{6}
	\gwall{0}{2}{0}{2}
	\gwall{3}{2}{8}{2}
	\outsidearea{0}{0}{8}{1}
	\gexit{1}{2}{2}{2}

	\gbackbone{2}{3}
	\gbackbone{3}{3}
	\gbackbone{4}{3}
	\gmodule{4}{4}
	\gbackbone{2}{1}
	\gbackbone{2}{2}
	\gmodule{1}{1}

	\draw[myarrow] (4.5,4.5) -- (2.2,4.5);
	\draw[myarrow] (1.5,1.5) -- (1.5,0.5) -- (2.8,0.5);
\end{customviewplot}
\ra
\begin{customviewplot}{9}{6}
	\gwall{0}{2}{0}{2}
	\gwall{3}{2}{8}{2}
	\outsidearea{0}{0}{8}{1}
	\gexit{1}{2}{2}{2}

	\gbackbone{2}{3}
	\gbackbone{3}{3}
	\gbackbone{4}{3}
	\gbackbone{2}{4}
	\gbackbone{2}{1}
	\gbackbone{2}{2}
	\gmodule{2}{0}

	\draw[myarrow] (2.5,0.5) -- (3.5,0.5) -- (3.5,1.8);
\end{customviewplot}
\ra
\begin{customviewplot}{9}{6}
	\gwall{0}{2}{0}{2}
	\gwall{3}{2}{8}{2}
	\outsidearea{0}{0}{8}{1}
	\gexit{1}{2}{2}{2}

	\gbackbone{2}{3}
	\gbackbone{3}{3}
	\gmodule{4}{3}
	\gmodule{2}{4}
	\gbackbone{2}{1}
	\gbackbone{2}{2}
	\gbackbone{3}{1}

	\draw[myarrow] (2.5,4.5) -- (1.5,4.5) -- (1.5,3.2);
	\draw[myarrow] (4.5,3.5) -- (4.5,4.5) -- (3.2,4.5);
\end{customviewplot}
\ra
\begin{customviewplot}{9}{6}
	\gwall{0}{2}{0}{2}
	\gwall{3}{2}{8}{2}
	\outsidearea{0}{0}{8}{1}
	\gexit{1}{2}{2}{2}

	\gbackbone{2}{3}
	\gbackbone{3}{3}
	\gbackbone{3}{4}
	\gmodule{1}{3}
	\gbackbone{2}{1}
	\gbackbone{2}{2}
	\gbackbone{3}{1}

	\draw[myarrow] (1.5,3.5) -- (1.5,1.2);
\end{customviewplot}
\ra
\\ \\
\begin{customviewplot}{9}{6}
	\gwall{0}{2}{0}{2}
	\gwall{3}{2}{8}{2}
	\outsidearea{0}{0}{8}{1}
	\gexit{1}{2}{2}{2}

	\gbackbone{2}{3}
	\gbackbone{3}{3}
	\gmodule{3}{4}
	\gmodule{1}{1}
	\gbackbone{2}{1}
	\gbackbone{2}{2}
	\gbackbone{3}{1}

	\draw[myarrow] (1.5,1.5) -- (1.5,0.5) -- (2.8,0.5);
	\draw[myarrow] (3.5,4.5) -- (2.2,4.5);
\end{customviewplot}
\ra
\begin{customviewplot}{9}{6}
	\gwall{0}{2}{0}{2}
	\gwall{3}{2}{8}{2}
	\outsidearea{0}{0}{8}{1}
	\gexit{1}{2}{2}{2}

	\gbackbone{2}{3}
	\gbackbone{3}{3}
	\gmodule{2}{4}
	\gmodule{2}{0}
	\gbackbone{2}{1}
	\gbackbone{2}{2}
	\gbackbone{3}{1}

	\draw[myarrow] (2.5,0.5) -- (3.8,0.5);
	\draw[myarrow] (2.5,4.5) -- (1.5,4.5) -- (1.5,3.2);
\end{customviewplot}
\ra
\begin{customviewplot}{9}{6}
	\gwall{0}{2}{0}{2}
	\gwall{3}{2}{8}{2}
	\outsidearea{0}{0}{8}{1}
	\gexit{1}{2}{2}{2}

	\gbackbone{2}{3}
	\gbackbone{3}{3}
	\gmodule{1}{3}
	\gmodule{3}{0}
	\gbackbone{2}{1}
	\gbackbone{2}{2}
	\gbackbone{3}{1}

	\draw[myarrow] (3.5,0.5) -- (4.5,0.5) -- (4.5,1.8);
	\draw[myarrow] (1.5,3.5) -- (1.5,1.2);
\end{customviewplot}
\ra
\begin{customviewplot}{9}{6}
	\gwall{0}{2}{0}{2}
	\gwall{3}{2}{8}{2}
	\outsidearea{0}{0}{8}{1}
	\gexit{1}{2}{2}{2}

	\gbackbone{2}{3}
	\gmodule{3}{3}
	\gmodule{1}{1}
	\gbackbone{4}{1}
	\gbackbone{2}{1}
	\gbackbone{2}{2}
	\gbackbone{3}{1}

	\draw[myarrow] (1.5,1.5) -- (1.5,0.5) -- (2.8,0.5);
	\draw[myarrow] (3.5,3.5) -- (3.5,4.5) -- (2.2,4.5);
\end{customviewplot}
\ra
\begin{customviewplot}{9}{6}
	\gwall{0}{2}{0}{2}
	\gwall{3}{2}{8}{2}
	\outsidearea{0}{0}{8}{1}
	\gexit{1}{2}{2}{2}

	\gbackbone{2}{3}
	\gmodule{2}{4}
	\gbackbone{2}{0}
	\gbackbone{4}{1}
	\gbackbone{2}{1}
	\gbackbone{2}{2}
	\gbackbone{3}{1}

	\draw[myarrow] (2.5,4.5) -- (1.5,4.5) -- (1.5,3.2);
\end{customviewplot}
\ra
\begin{customviewplot}{9}{6}
	\gwall{0}{2}{0}{2}
	\gwall{3}{2}{8}{2}
	\outsidearea{0}{0}{8}{1}
	\gexit{1}{2}{2}{2}

	\gbackbone{2}{3}
	\gbackbone{1}{3}
	\gmodule{2}{0}
	\gbackbone{4}{1}
	\gbackbone{2}{1}
	\gbackbone{2}{2}
	\gbackbone{3}{1}

	\draw[myarrow] (2.5,0.5) -- (4.8,0.5);
\end{customviewplot}
\ra
\\ \\
\begin{customviewplot}{9}{6}
	\gwall{0}{2}{0}{2}
	\gwall{3}{2}{8}{2}
	\outsidearea{0}{0}{8}{1}
	\gexit{1}{2}{2}{2}

	\gbackbone{2}{3}
	\gmodule{1}{3}
	\gmodule{4}{0}
	\gbackbone{4}{1}
	\gbackbone{2}{1}
	\gbackbone{2}{2}
	\gbackbone{3}{1}

	\draw[myarrow] (1.5,3.5) -- (1.5,1.2);
	\draw[myarrow] (4.5,0.5) -- (5.5,0.5) -- (5.5,1.8);
\end{customviewplot}
\ra
\begin{customviewplot}{9}{6}
	\gwall{0}{2}{0}{2}
	\gwall{3}{2}{8}{2}
	\outsidearea{0}{0}{8}{1}
	\gexit{1}{2}{2}{2}

	\gmodule{2}{3}
	\gmodule{1}{1}
	\gbackbone{5}{1}
	\gbackbone{4}{1}
	\gbackbone{2}{1}
	\gbackbone{2}{2}
	\gbackbone{3}{1}

	\draw[myarrow] (1.5,1.5) -- (1.5,0.5) -- (2.8,0.5);
	\draw[myarrow] (2.5,3.5) -- (1.5,3.5) -- (1.5,2.2);
\end{customviewplot}
\ra
\begin{customviewplot}{9}{6}
	\gwall{0}{2}{0}{2}
	\gwall{3}{2}{8}{2}
	\outsidearea{0}{0}{8}{1}
	\gexit{1}{2}{2}{2}

	\gbackbone{1}{2}
	\gmodule{2}{0}
	\gbackbone{5}{1}
	\gbackbone{4}{1}
	\gbackbone{2}{1}
	\gbackbone{2}{2}
	\gbackbone{3}{1}

	\draw[myarrow] (2.5,0.5) -- (5.8,0.5);
\end{customviewplot}
\ra
\begin{customviewplot}{9}{6}
	\gwall{0}{2}{0}{2}
	\gwall{3}{2}{8}{2}
	\outsidearea{0}{0}{8}{1}
	\gexit{1}{2}{2}{2}

	\gmodule{1}{2}
	\gmodule{5}{0}
	\gbackbone{5}{1}
	\gbackbone{4}{1}
	\gbackbone{2}{1}
	\gbackbone{2}{2}
	\gbackbone{3}{1}

	\draw[myarrow] (5.5,0.5) -- (6.5,0.5) -- (6.5,1.8);
	\draw[myarrow] (1.5,2.5) -- (1.5,1.2);
\end{customviewplot}
\ra
\begin{customviewplot}{9}{6}
	\gwall{0}{2}{0}{2}
	\gwall{3}{2}{8}{2}
	\outsidearea{0}{0}{8}{1}
	\gexit{1}{2}{2}{2}

	\gmodule{1}{1}
	\gbackbone{6}{1}
	\gbackbone{5}{1}
	\gbackbone{4}{1}
	\gbackbone{2}{1}
	\gbackbone{2}{2}
	\gbackbone{3}{1}

	\draw[myarrow] (1.5,1.5) -- (1.5,0.5) -- (2.8,0.5);
\end{customviewplot}
\ra
\begin{customviewplot}{9}{6}
	\gwall{0}{2}{0}{2}
	\gwall{3}{2}{8}{2}
	\outsidearea{0}{0}{8}{1}
	\gexit{1}{2}{2}{2}

	\gmodule{2}{0}
	\gbackbone{6}{1}
	\gbackbone{5}{1}
	\gbackbone{4}{1}
	\gbackbone{2}{1}
	\gbackbone{2}{2}
	\gbackbone{3}{1}

	\draw[myarrow] (2.5,0.5) -- (6.8,0.5);
\end{customviewplot}
\ra
\\ \\
\begin{customviewplot}{9}{6}
	\gwall{0}{2}{0}{2}
	\gwall{3}{2}{8}{2}
	\outsidearea{0}{0}{8}{1}
	\gexit{1}{2}{2}{2}

	\gbackbone{6}{0}
	\gbackbone{6}{1}
	\gbackbone{5}{1}
	\gbackbone{4}{1}
	\gbackbone{2}{1}
	\gmodule{2}{2}
	\gbackbone{3}{1}

	\draw[myarrow] (2.5,2.5) -- (1.5,2.5) -- (1.5,1.2);
\end{customviewplot}
\ra
\begin{customviewplot}{9}{6}
	\gwall{0}{2}{0}{2}
	\gwall{3}{2}{8}{2}
	\outsidearea{0}{0}{8}{1}
	\gexit{1}{2}{2}{2}

	\gmodule{6}{0}
	\gbackbone{6}{1}
	\gbackbone{5}{1}
	\gbackbone{4}{1}
	\gbackbone{2}{1}
	\gmodule{1}{1}
	\gbackbone{3}{1}

	\draw[myarrow] (1.5,1.5) -- (1.5,0.5) -- (2.8,0.5);
	\draw[myarrow] (6.5,0.5) -- (7.5,0.5) -- (7.5,1.8);
\end{customviewplot}
\ra
\begin{customviewplot}{9}{6}
	\gwall{0}{2}{0}{2}
	\gwall{3}{2}{8}{2}
	\outsidearea{0}{0}{8}{1}
	\gexit{1}{2}{2}{2}

	\gbackbone{7}{1}
	\gbackbone{6}{1}
	\gbackbone{5}{1}
	\gbackbone{4}{1}
	\gbackbone{2}{1}
	\gbackbone{3}{1}
	\gmodule{2}{0}

	\draw[myarrow] (2.5,0.5) -- (7.8,0.5);
\end{customviewplot}
\ra
\begin{customviewplot}{9}{6}
	\gwall{0}{2}{0}{2}
	\gwall{3}{2}{8}{2}
	\outsidearea{0}{0}{8}{1}
	\gexit{1}{2}{2}{2}

	\gmodule{7}{0}
	\gbackbone{7}{1}
	\gbackbone{6}{1}
	\gbackbone{5}{1}
	\gbackbone{3}{1}
	\gmodule{2}{1}
	\gbackbone{4}{1}

	\draw[myarrow] (2.5,1.5) -- (2.5,0.5) -- (3.8,0.5);
	\draw[myarrow] (7.5,0.5) -- (8.5,0.5) -- (8.5,1.8);
\end{customviewplot}
\ra
\begin{customviewplot}{9}{6}
	\gwall{0}{2}{0}{2}
	\gwall{3}{2}{8}{2}
	\outsidearea{0}{0}{8}{1}
	\gexit{1}{2}{2}{2}

	\gbackbone{8}{1}
	\gbackbone{7}{1}
	\gbackbone{6}{1}
	\gbackbone{5}{1}
	\gbackbone{3}{1}
	\gbackbone{4}{1}
	\gmodule{3}{0}

	\draw[myarrow] (3.5,0.5) -- (8.8,0.5);
\end{customviewplot}
\ra\ $\cdots$
\end{tabular}

	\caption{Evacuation through an exit by seven modules without a global compass}
	\label{fig:evacuate-from-exit-wo-gc-seven}
\end{figure}

This locomotion algorithm occupies five cells in the same row in the second and third states in Fig.~\ref{fig:locomotion-with-seven-modules}.
This implies that the width and height of a rectangular field must be at least five cells each.
Therefore, fields to which this evacuation algorithm is applicable are limited as follows.
\begin{theorem}
\label{thm:evacuation-without-gc-arbitrary-initial-states}
Seven modules not equipped with a global compass are necessary and sufficient for a metamorphic robotic system to solve the evacuation problem starting from an arbitrary initial state in any rectangular field whose width and height are at least five cells each.
\end{theorem}
\noindent
This evacuation algorithm requires a visibility range of $11 \times 11$ (i.e., 5-neighborhood).
This range is necessary to avoid collisions of modules when a module performs a 5-sliding movement, as in the example shown in Fig.~\ref{fig:move-along-wall-wo-gc-seven}.

\subsubsection{Starting on a wall}
\label{sec:evacuation-without-gc-along-the-wall}

Here, we consider evacuation when an MRS composed of modules not equipped with a global compass starts on a wall.
Contrary to Sections \ref{sec:evacuation-without-gc-restricted-init-states} and \ref{sec:evacuation-without-gc-arbitrary-init-states}, two modules are necessary and sufficient for evacuation in this case.

The movements of the two modules are similar to those of two modules equipped with a global compass (Figs.~\ref{fig:move-along-wall-with-gc}--\ref{fig:evacuate-from-exit-with-gc}), even when a global compass is unavailable.
This is possible because starting on a wall allows the modules to agree on a common direction of movement, which acts as a substitute for a global compass. This is accomplished by making use of the direction in which wall cells are observed in their views, as follows.
At each step, a module $m$ observes neighborhood cells and obtains its view $v$.
If a module is side-adjacent to a wall cell in a direction $d$ in view $v$, $m$ determines its direction of movement as $d + \pi/2$ in a clockwise direction.
For example, if a module $m$ observes based on its local compass that its neighbor module is side-adjacent to a wall cell in the north, $m$ changes direction of movement to the east.
This decision is based only on the handedness shared among the modules; thus, a global compass is not required.
The visibility range of this algorithm are the same as that in Section \ref{sec:evacuation-rectangular-with-global-compass}.

However, this method has a limitation on the field size.
Let us consider the cases depicted in Fig.~\ref{fig:field-size-limitation-wo-global-compass-two-module}.
In Fig.~\ref{fig:field-size-limitation-wo-global-compass-two-module}(a), modules $a$ and $b$ may have the same view in an initial state.
If module $a$ (resp.~$b$) recognizes the right (resp.~left) direction as the north, module $a$ (resp.~$b$) waits for $b$ (resp.~$a$) to rotate.
Then, both modules do not move.
As a result, the MRS cannot move.
On the other hand, in the case shown in Fig.~\ref{fig:field-size-limitation-wo-global-compass-two-module}(b), the views of the two modules are different.
Therefore, $b$ becomes a backbone module and $a$ rotates around $b$.
This allows the MRS to move along the walls in a clockwise direction.
\begin{figure}[t]
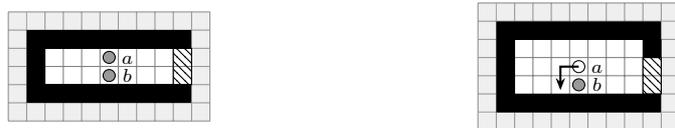

	\centering
	\begin{tabular}{cc}

	\begin{customviewplot}{11}{6}
		\gwall{1}{1}{9}{1}
		\gwall{1}{4}{9}{4}
		\gwall{1}{2}{1}{3}
		\outside{1}
		\gexit{9}{2}{9}{3}

		\gbackbone{5}{2}
		\node at (6.5, 3.4) {{\scriptsize $\strut a$}};

		\gbackbone{5}{3}
		\node at (6.5, 2.4) {{\scriptsize $\strut b$}};

	\end{customviewplot}
	&
	\begin{customviewplot}{11}{7}
		\gwall{1}{1}{9}{1}
		\gwall{1}{5}{9}{5}
		\gwall{1}{2}{1}{4}
		\gwall{9}{4}{9}{4}
		\outside{1}
		\gexit{9}{2}{9}{3}

		\gmodule{5}{3}
		\node at (6.5, 3.4) {{\scriptsize $\strut a$}};

		\gbackbone{5}{2}
		\node at (6.5, 2.4) {{\scriptsize $\strut b$}};

		\draw[myarrow] (5.5,3.5) -- (4.5,3.5) -- (4.5,2.2);

	\end{customviewplot}
	\\
	\subcaption{(a) Both the modules have the same view and stop forever}
	&
	\subcaption{(b) Correct movement}

\end{tabular}

	\caption{Field size limitation in the case of two modules without a global compass }
	\label{fig:field-size-limitation-wo-global-compass-two-module}
\end{figure}
To avoid this case, the following theorem restricts the size of a rectangular field.
\begin{theorem}
\label{thm:evacuation-without-gc-along-the-wall}
Two modules not equipped with a global compass are necessary and sufficient for a metamorphic robotic system to solve the evacuation problem from a side of a wall in any rectangular field whose width and height are at least three cells.
\end{theorem}

\subsection{Stopping after evacuation}
\label{sec:stop-after-evacuation}

In this section, we discuss the possibility of stopping an MRS after it evacuates through an exit from a rectangular field.
First, we show the following theorem regarding the positive result to stop an MRS from moving after evacuation.

\begin{theorem}
\label{thm:termination-possible-if-only-one-exit-exists}
An MRS composed of two or more modules not equipped with a global compass can stop after evacuation if a rectangular field has only one exit.
\end{theorem}
\begin{proof}
After evacuating from an exit, an MRS eventually reaches a corner by moving along walls of the exterior side.
Figure \ref{fig:unique-state-after-evacuation} shows such situations.
In these situations, each module can detect that an MRS is in the exterior\footnote{Recall that each module cannot distinguish an interior cell and an exterior cell.} because this positional relationship between an MRS and the walls appears only after evacuation.
Note that this detection does not requires a global compass.
Therefore, each module can stop moving when it observes the views of Fig.~\ref{fig:unique-state-after-evacuation}.
The views are carefully chosen based on their visibility range so that all the modules decide to stop moving in the same step, which is necessary not to break the connectivity of modules.
\begin{figure}[t]
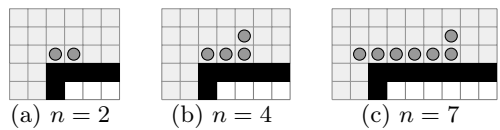

	\centering
	\begin{tabular}{ccc}
\begin{customviewplot}{6}{5}
	\outsidearea{0}{0}{1}{1}
	\outsidearea{0}{2}{5}{4}

	\gwall{2}{0}{2}{1}
	\gwall{2}{1}{5}{1}

	\gbackbone{2}{2}
	\gbackbone{3}{2}
\end{customviewplot}

&

\begin{customviewplot}{7}{5}
	\outsidearea{0}{0}{1}{1}
	\outsidearea{0}{2}{6}{4}

	\gwall{2}{0}{2}{1}
	\gwall{2}{1}{6}{1}

	\gbackbone{2}{2}
	\gbackbone{3}{2}
	\gbackbone{4}{2}
	\gbackbone{4}{3}

\end{customviewplot}

&

\begin{customviewplot}{9}{5}
	\outsidearea{0}{0}{1}{1}
	\outsidearea{0}{2}{8}{4}

	\gwall{2}{0}{2}{1}
	\gwall{2}{1}{8}{1}

	\gbackbone{1}{2}
	\gbackbone{2}{2}
	\gbackbone{3}{2}
	\gbackbone{4}{2}
	\gbackbone{5}{2}
	\gbackbone{6}{2}
	\gbackbone{6}{3}

\end{customviewplot}
\\

\subcaption{(a) $n=2$} &
\subcaption{(b) $n=4$} &
\subcaption{(c) $n=7$} \\

\end{tabular}

	\caption{Termination states after evacuation}
	\label{fig:unique-state-after-evacuation}
\end{figure}
\end{proof}

However, it is impossible to stop after evacuation if a rectangular filed has two or more exits and there is no assumption on an initial state of an MRS.

\begin{theorem}
\label{thm:mrs-keep-moving-after-evacuation}
Let $R$ be an MRS composed of two or more modules equipped with a global compass and having limited visibility.
Let $F$ be a rectangular field that has two or more exits.
There is no algorithm that can solve the evacuation problem from any initial state and can stop $R$ from moving after evacuation through an exit of $F$.
\end{theorem}
\begin{proof}
Note that there is a situation in which an MRS cannot reach the terminal state of Theorem \ref{thm:termination-possible-if-only-one-exit-exists} because $F$ has two exits.
In the situation shown in Fig. \ref{fig:counterexample-to-stop-at-corner}, the MRS goes east first until reaching a wall.
After that, it moves along the wall and evacuates from exit $a$.
Then, the MRS still moves along the wall.
However, it finds exit $b$ before observing a corner and enters the interior again through $b$.
The MRS repeats these movements forever and never stops moving.
Therefore, we cannot take this approach based on the views of a corner.
\begin{figure}[t]
       \centering
       \input{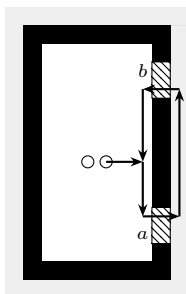}
       \caption{Counterexample to stop moving at a corner}
       \label{fig:counterexample-to-stop-at-corner}
\end{figure}

For contradiction, we assume that there exists an algorithm $\mathcal{A}$ that can solve the evacuation problem from any initial state and can stop an MRS from moving after completing evacuation.
Let $F$ be a rectangular field sufficiently large so that no module can see any corner from any exit (Fig.~\ref{fig:impossible-termination-counterexample}(a)).
Let us consider evacuation from the rectangular field $F$.
Without loss of generality, we assume that $F$ has two exits $\{ c_{a,-1},c_{a+1,-1} \}$ and $\{ c_{b,-1},c_{b+1,-1} \}$ in the south wall.
Because $\mathcal{A}$ is a deterministic algorithm, the terminal state of MRS $R$ is uniquely determined from the initial position and shape of $R$.
Let us assume that the terminal state of MRS $R$ reached by $\mathcal{A}$ from a given initial position and shape is as shown in Fig.~\ref{fig:impossible-termination-counterexample}(a).
Figure \ref{fig:impossible-termination-counterexample}(b) illustrates the views of the modules.

Assume another rectangular field $F'$ whose size is the same as $F$ (Fig.~\ref{fig:impossible-termination-counterexample}(c)).
Unlike $F$, $F'$ has two exits $\{ c_{a,h},c_{a+1,h} \}$ and $\{ c_{b,h},c_{b+1,h} \}$ on the north wall.
Let us consider an MRS $R'$ starts its evacuation in $F'$ from the state in which each module has an identical view with that of the terminal state of $F$.
Figure \ref{fig:impossible-termination-counterexample}(d) depicts the views of modules in $F'$.
Because the views of the modules in $F'$ are identical to $F$, no module in $F'$ moves.
Note that a module cannot distinguish between a cell in the exterior and one in the interior.
Therefore, the algorithm $\mathcal{A}$ cannot solve the evacuation problem in the rectangular field $F'$.
This contradicts the assumption.
\begin{figure}[t]
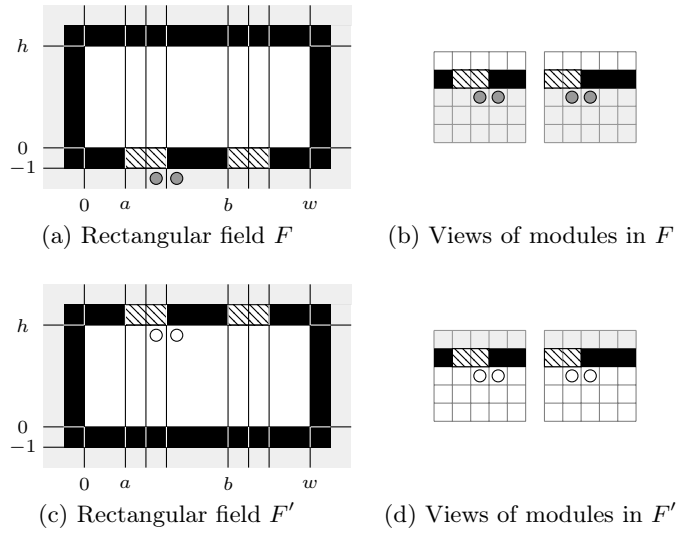

	\centering
	\begin{tabular}{cc}

\setlength{\gridstep}{2.7mm}
\begin{nogridplot}{15}{9}
\outside{1}

\gwall{1}{1}{1}{7}
\gwall{13}{1}{13}{7}

\gwall{1}{7}{13}{7}

\gwall{1}{1}{3}{1}
\gwall{11}{1}{13}{1}
\gwall{6}{1}{8}{1}

\gexit{4}{1}{5}{1}
\gexit{9}{1}{10}{1}

\node at (-1.0,2.0) {\scriptsize $0$};
\node at (-1.0,1.0) {\scriptsize $-1$};
\node at (-1.0,7.0) {\scriptsize $h$};

\node at (2.0,-0.8) {\scriptsize $0$};
\node at (4.0,-0.8) {\scriptsize $a$};
\node at (9.0,-0.8) {\scriptsize $b$};
\node at (13.0,-0.8) {\scriptsize $w$};

\draw [black] (0.0,1.0) -- (14.0,1.0);
\draw [black] (0.0,2.0) -- (1.0,2.0);
\draw [white] (1.0,2.0) -- (2.0,2.0);
\draw [black] (2.0,2.0) -- (13.0,2.0);
\draw [white] (13.0,2.0) -- (14.0,2.0);
\draw [black] (14.0,2.0) -- (15.0,2.0);
\draw [black] (0.0,7.0) -- (1.0,7.0);
\draw [white] (1.0,7.0) -- (2.0,7.0);
\draw [black] (2.0,7.0) -- (13.0,7.0);
\draw [white] (13.0,7.0) -- (14.0,7.0);
\draw [black] (14.0,7.0) -- (15.0,7.0);

\draw [black] (2.0,0.0) -- (2.0,1.0);
\draw [white] (2.0,1.0) -- (2.0,2.0);
\draw [black] (2.0,2.0) -- (2.0,7.0);
\draw [white] (2.0,7.0) -- (2.0,8.0);
\draw [black] (2.0,8.0) -- (2.0,9.0);
\draw [black] (4.0,0.0) -- (4.0,1.0);
\draw [white] (4.0,1.0) -- (4.0,2.0);
\draw [black] (4.0,2.0) -- (4.0,7.0);
\draw [white] (4.0,7.0) -- (4.0,8.0);
\draw [black] (4.0,8.0) -- (4.0,9.0);
\draw [black] (5.0,0.0) -- (5.0,7.0);
\draw [white] (5.0,7.0) -- (5.0,8.0);
\draw [black] (5.0,8.0) -- (5.0,9.0);
\draw [black] (6.0,0.0) -- (6.0,1.0);
\draw [white] (6.0,1.0) -- (6.0,2.0);
\draw [black] (6.0,2.0) -- (6.0,7.0);
\draw [white] (6.0,7.0) -- (6.0,8.0);
\draw [black] (6.0,8.0) -- (6.0,9.0);
\draw [black] (9.0,0.0) -- (9.0,1.0);
\draw [white] (9.0,1.0) -- (9.0,2.0);
\draw [black] (9.0,2.0) -- (9.0,7.0);
\draw [white] (9.0,7.0) -- (9.0,8.0);
\draw [black] (9.0,8.0) -- (9.0,9.0);
\draw [black] (10.0,0.0) -- (10.0,7.0);
\draw [white] (10.0,7.0) -- (10.0,8.0);
\draw [black] (10.0,8.0) -- (10.0,9.0);
\draw [black] (11.0,0.0) -- (11.0,1.0);
\draw [white] (11.0,1.0) -- (11.0,2.0);
\draw [black] (11.0,2.0) -- (11.0,7.0);
\draw [white] (11.0,7.0) -- (11.0,8.0);
\draw [black] (11.0,8.0) -- (11.0,9.0);
\draw [black] (13.0,0.0) -- (13.0,1.0);
\draw [white] (13.0,1.0) -- (13.0,2.0);
\draw [black] (13.0,2.0) -- (13.0,7.0);
\draw [white] (13.0,7.0) -- (13.0,8.0);
\draw [black] (13.0,8.0) -- (13.0,9.0);

\gbackbone{5}{0}
\gbackbone{6}{0}

\end{nogridplot}

&

\begin{customviewplot}{5}{5}
	\gwall{0}{3}{0}{3}
	\gwall{3}{3}{4}{3}
	\outsidearea{0}{0}{4}{2}
	\gexit{1}{3}{2}{3}

	\gbackbone{2}{2}
	\gbackbone{3}{2}
\end{customviewplot}

\begin{customviewplot}{5}{5}
	\gwall{2}{3}{4}{3}
	\outsidearea{0}{0}{4}{2}
	\gexit{0}{3}{1}{3}

	\gbackbone{1}{2}
	\gbackbone{2}{2}
\end{customviewplot}

\\
\subcaption{(a) Rectangular field $F$} & \subcaption{(b) Views of modules in $F$} \\ \\

\setlength{\gridstep}{2.7mm}
\begin{nogridplot}{15}{9}
\outside{1}

\gwall{1}{1}{1}{7}
\gwall{13}{1}{13}{7}

\gwall{1}{1}{13}{1}

\gwall{1}{7}{3}{7}
\gwall{11}{7}{13}{7}
\gwall{6}{7}{8}{7}

\gexit{4}{7}{5}{7}
\gexit{9}{7}{10}{7}

\gmodule{5}{6}
\gmodule{6}{6}

\node at (-1.0,2.0) {\scriptsize $0$};
\node at (-1.0,1.0) {\scriptsize $-1$};
\node at (-1.0,7.0) {\scriptsize $h$};

\node at (2.0,-0.8) {\scriptsize $0$};
\node at (4.0,-0.8) {\scriptsize $a$};
\node at (9.0,-0.8) {\scriptsize $b$};
\node at (13.0,-0.8) {\scriptsize $w$};

\draw [black] (0.0,1.0) -- (14.0,1.0);
\draw [black] (0.0,2.0) -- (1.0,2.0);
\draw [white] (1.0,2.0) -- (2.0,2.0);
\draw [black] (2.0,2.0) -- (13.0,2.0);
\draw [white] (13.0,2.0) -- (14.0,2.0);
\draw [black] (14.0,2.0) -- (15.0,2.0);
\draw [black] (0.0,7.0) -- (1.0,7.0);
\draw [white] (1.0,7.0) -- (2.0,7.0);
\draw [black] (2.0,7.0) -- (13.0,7.0);
\draw [white] (13.0,7.0) -- (14.0,7.0);
\draw [black] (14.0,7.0) -- (15.0,7.0);

\draw [black] (2.0,0.0) -- (2.0,1.0);
\draw [white] (2.0,1.0) -- (2.0,2.0);
\draw [black] (2.0,2.0) -- (2.0,7.0);
\draw [white] (2.0,7.0) -- (2.0,8.0);
\draw [black] (2.0,8.0) -- (2.0,9.0);
\draw [black] (4.0,0.0) -- (4.0,1.0);
\draw [white] (4.0,1.0) -- (4.0,2.0);
\draw [black] (4.0,2.0) -- (4.0,7.0);
\draw [white] (4.0,7.0) -- (4.0,8.0);
\draw [black] (4.0,8.0) -- (4.0,9.0);
\draw [black] (5.0,0.0) -- (5.0,1.0);
\draw [white] (5.0,1.0) -- (5.0,2.0);
\draw [black] (5.0,2.0) -- (5.0,9.0);
\draw [black] (6.0,0.0) -- (6.0,1.0);
\draw [white] (6.0,1.0) -- (6.0,2.0);
\draw [black] (6.0,2.0) -- (6.0,7.0);
\draw [white] (6.0,7.0) -- (6.0,8.0);
\draw [black] (6.0,8.0) -- (6.0,9.0);
\draw [black] (9.0,0.0) -- (9.0,1.0);
\draw [white] (9.0,1.0) -- (9.0,2.0);
\draw [black] (9.0,2.0) -- (9.0,7.0);
\draw [white] (9.0,7.0) -- (9.0,8.0);
\draw [black] (9.0,8.0) -- (9.0,9.0);
\draw [black] (10.0,0.0) -- (10.0,1.0);
\draw [white] (10.0,1.0) -- (10.0,2.0);
\draw [black] (10.0,2.0) -- (10.0,9.0);
\draw [black] (11.0,0.0) -- (11.0,1.0);
\draw [white] (11.0,1.0) -- (11.0,2.0);
\draw [black] (11.0,2.0) -- (11.0,7.0);
\draw [white] (11.0,7.0) -- (11.0,8.0);
\draw [black] (11.0,8.0) -- (11.0,9.0);
\draw [black] (13.0,0.0) -- (13.0,1.0);
\draw [white] (13.0,1.0) -- (13.0,2.0);
\draw [black] (13.0,2.0) -- (13.0,7.0);
\draw [white] (13.0,7.0) -- (13.0,8.0);
\draw [black] (13.0,8.0) -- (13.0,9.0);

\end{nogridplot}

&

\begin{customviewplot}{5}{5}
	\gwall{0}{3}{0}{3}
	\gwall{3}{3}{4}{3}
	\outsidearea{0}{4}{4}{4}
	\gexit{1}{3}{2}{3}

	\gmodule{2}{2}
	\gmodule{3}{2}
\end{customviewplot}

\begin{customviewplot}{5}{5}
	\gwall{2}{3}{4}{3}
	\outsidearea{0}{4}{4}{4}
	\gexit{0}{3}{1}{3}

	\gmodule{1}{2}
	\gmodule{2}{2}
\end{customviewplot}
\\
\subcaption{(c) Rectangular field $F'$} & \subcaption{(d) Views of modules in $F'$} \\

\end{tabular}
	\caption{Two rectangular fields assumed in the proof of Theorem \ref{thm:mrs-keep-moving-after-evacuation}. Here, we assume that $n=2$ and that modules have 2-neighborhood visibility.}
	\label{fig:impossible-termination-counterexample}
\end{figure}
\end{proof}

Only one exception exists for this result.
The algorithm in Section \ref{sec:evacuation-without-gc-restricted-init-states} can stop an MRS from moving after evacuation, even if a rectangular field has two or more exits.
This does not contradict Theorem \ref{thm:mrs-keep-moving-after-evacuation}, because the algorithm has the additional assumption on initial shapes of an MRS.
There are four forbidden initial shapes of an MRS (Fig.~\ref{fig:evacuation-without-compass-forbidden-states}), and we can utilize one of them as a terminal shape after evacuation.
In other words, none of the forbidden shapes is given to the algorithm as an initial state; thus, the algorithm does not need to consider evacuation from them.
From this discussion, we have the following theorem.
\begin{theorem}
\label{thm:mrs-can-stop-if-intial-shapes-are-restricted}
An MRS can stop moving after evacuation from a rectangular field if initial shapes of an MRS are restricted.
\end{theorem}

\section{Other types of fields}
\label{sec:other-types-of-field}

By extending the movements shown in Sections \ref{sec:evacuation-rectangular-with-global-compass} and \ref{sec:evacuation-rectangular-wo-global-compass}, an MRS can evacuate from more complex fields.
As examples, we now describe evacuation algorithms for a maze and a convex field.
The algorithms introduced here are applicable to various models, regardless of the existence of a global compass and restrictions on the initial position.
The number of modules and visibility range required to solve the evacuation problem are the same as those in rectangular fields.

\subsection{Evacuation from a maze}
\label{sec:evacuation-from-maze}

Here, we consider evacuation from mazes.
We defined a maze as a set of rectangular fields interconnected by exits in Definition \ref{def:maze}.
With this definition, we can apply evacuation algorithms for rectangular fields directly to mazes. Only the evacuation algorithm for four modules that are not equipped with a global compass (Section \ref{sec:evacuation-without-gc-arbitrary-init-states}) should be modified so that it does not stop an MRS after completing evacuation.

Figure \ref{fig:mrs-trail-in-maze} shows an example evacuation trail of an MRS composed of two modules equipped with a global compass.
\begin{figure}[t]
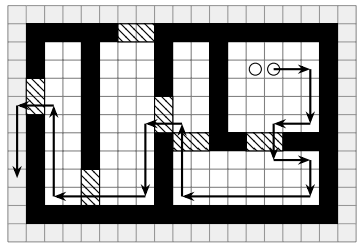

	\centering
	\begin{customviewplot}{19}{13}
\outside{1}

\gwall{1}{1}{17}{1}
\gwall{1}{11}{5}{11}
\gwall{8}{11}{17}{11}
\gwall{1}{2}{1}{6}
\gwall{1}{9}{1}{10}
\gwall{17}{2}{17}{10}

\gwall{4}{4}{4}{10}
\gwall{8}{2}{8}{5}
\gwall{8}{8}{8}{10}
\gwall{11}{5}{11}{10}
\gwall{12}{5}{12}{5}
\gwall{15}{5}{16}{5}

\gexit{1}{7}{1}{8}
\gexit{4}{2}{4}{3}
\gexit{8}{6}{8}{7}
\gexit{9}{5}{10}{5}
\gexit{13}{5}{14}{5}

\gexit{6}{11}{7}{11}

\gmodule{13}{9}
\gmodule{14}{9}

\draw[myarrow] (14.5,9.5) -- (16.5,9.5);
\draw[myarrow] (16.5,9.5) -- (16.5,6.5);
\draw[myarrow] (16.5,6.5) -- (14.5,6.5);
\draw[myarrow] (14.5,6.5) -- (14.5,4.5);
\draw[myarrow] (14.5,4.5) -- (16.5,4.5);
\draw[myarrow] (16.5,4.5) -- (16.5,2.5);
\draw[myarrow] (16.5,2.5) -- (9.5,2.5);
\draw[myarrow] (9.5,2.5) -- (9.5,6.5);
\draw[myarrow] (9.5,6.5) -- (7.5,6.5);
\draw[myarrow] (7.5,6.5) -- (7.5,2.5);
\draw[myarrow] (7.5,2.5) -- (2.5,2.5);
\draw[myarrow] (2.5,2.5) -- (2.5,7.5);
\draw[myarrow] (2.5,7.5) -- (0.5,7.5);
\draw[myarrow] (0.5,7.5) -- (0.5,3.5);

\end{customviewplot}
	\caption{Example evacuation trail of an MRS from a maze}
	\label{fig:mrs-trail-in-maze}
\end{figure}
First, an MRS reaches a wall from its initial position and shape by locomotion and finds an exit by moving along the walls.
Then it passes through that exit and enters another rectangular field adjacent to the current field, from which it finds an exit leading to the next field.
By repeating these movements, an MRS eventually finds an exit to the exterior and evacuates from the maze.
This approach simulates the well-known wall follower method,\footnote{\url{https://en.wikipedia.org/wiki/Maze_solving_algorithm\#Wall_follower}} and the correctness of our approach follows this method.

Contrary to the results in Section \ref{sec:stop-after-evacuation}, it is impossible for an MRS to stop moving after evacuation from a maze, even if a maze has only one exit or we restrict the initial shapes of an MRS.
Figure \ref{fig:counterexample-to-stop-in-maze} shows counterexamples for both cases, in which an MRS stops moving by mistake after passing to the second field through an exit of the first field.
For the case of only one exit (Fig.~\ref{fig:counterexample-to-stop-in-maze}(a)), an MRS stops after reaching the corner of the first field because each module has the same view as that of the terminal state.
For the case of the restriction of the initial shapes (Fig.~\ref{fig:counterexample-to-stop-in-maze}(b)), an MRS immediately stops after the passing with one of the forbidden shapes.
This impossibility comes from the fact that an MRS cannot distinguish a final exit to the exterior from other exits that interconnect two interiors.
\begin{figure}[t]
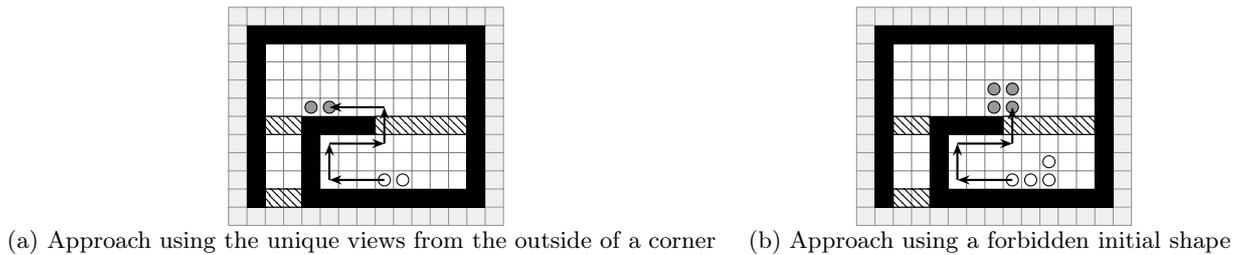

	\centering
	\begin{tabular}{cc}

\begin{customviewplot}{15}{12}
\outside{1}

\gwall{1}{1}{1}{10}
\gwall{13}{1}{13}{10}
\gwall{1}{10}{13}{10}
\gwall{4}{1}{13}{1}
\gwall{4}{1}{4}{5}
\gwall{4}{5}{7}{5}

\gexit{2}{1}{3}{1}
\gexit{2}{5}{3}{5}
\gexit{8}{5}{12}{5}

\gmodule{8}{2}
\gmodule{9}{2}

\gbackbone{4}{6}
\gbackbone{5}{6}

\draw[myarrow] (8.5,2.5) -- (5.5,2.5);
\draw[myarrow] (5.5,2.5) -- (5.5,4.5);
\draw[myarrow] (5.5,4.5) -- (8.5,4.5);
\draw[myarrow] (8.5,4.5) -- (8.5,6.5);
\draw[myarrow] (8.5,6.5) -- (5.5,6.5);

\end{customviewplot}

&

\begin{customviewplot}{15}{12}
\outside{1}

\gwall{1}{1}{1}{10}
\gwall{13}{1}{13}{10}
\gwall{1}{10}{13}{10}
\gwall{4}{1}{13}{1}
\gwall{4}{1}{4}{5}
\gwall{4}{5}{7}{5}

\gexit{2}{1}{3}{1}
\gexit{2}{5}{3}{5}
\gexit{8}{5}{12}{5}

\gmodule{8}{2}
\gmodule{9}{2}
\gmodule{10}{2}
\gmodule{10}{3}

\gbackbone{7}{6}
\gbackbone{8}{6}
\gbackbone{7}{7}
\gbackbone{8}{7}

\draw[myarrow] (8.5,2.5) -- (5.5,2.5);
\draw[myarrow] (5.5,2.5) -- (5.5,4.5);
\draw[myarrow] (5.5,4.5) -- (8.5,4.5);
\draw[myarrow] (8.5,4.5) -- (8.5,6.5);

\end{customviewplot}

\\

\subcaption{(a) Approach using the unique views from the outside of a corner}
&
\subcaption{(b) Approach using a forbidden initial shape}
\\

\end{tabular}

	\caption{Counterexamples to stop moving after evacuation from a maze}
	\label{fig:counterexample-to-stop-in-maze}
\end{figure}

\subsection{Evacuation from a convex field}
\label{sec:evacuation-from-convex-field}

Here, we consider evacuation from convex fields.
Convex fields introduce stairs, which are a wall feature not found in rectangular fields.
Examples of the shape are shown in the bottom of Fig.~\ref{fig:example-convex-fields}(b), in which the bottom left and right walls both include stairs leading toward or away from the exit.
Thus, solving the evacuation problem from a convex field requires a type of movement which can ascend or descend stairs.
Figures \ref{fig:convex-moves-2-modules}, \ref{fig:convex-moves-4-modules}, and \ref{fig:convex-moves-7-modules} show the movements required to go up or down a step for an MRS of 2, 4, and 7 modules, respectively.
The dotted area in the figures means that each module does not care whether a cell of the area in its view is empty, a wall, or an exit.
Note that, for the down movements, an MRS moves horizontally by the movements shown in Figs.~\ref{fig:move-along-wall-with-gc}, \ref{fig:evacuation-without-gc-move-along-wall}, and \ref{fig:move-along-wall-wo-gc-seven} until it reaches the first position of Figs.~\ref{fig:convex-moves-2-modules}(b), \ref{fig:convex-moves-4-modules}(b), and \ref{fig:convex-moves-7-modules}(b); then it goes down a step.
Repeating these movements allows an MRS to go up or down stairs.
With the combination of these new movements and the ones used for a rectangular field, an MRS can move along walls in the interior of a convex field and find an exit.
Figure \ref{fig:mrs-trail-in-convex-field} is an example showing the trail of an MRS consisting of two modules with a global compass performing these movements in a convex field.

Let us discuss the possibility that an MRS stops after evacuation from a convex field.
For the case of the restriction of initial shapes, an MRS can stop after evacuation with the same discussion to Theorem \ref{thm:mrs-can-stop-if-intial-shapes-are-restricted}.
However, it is impossible to take the same approach as Theorem \ref{thm:termination-possible-if-only-one-exit-exists} because we do not have such unique views from the exterior of every convex field.

\begin{figure}[t]
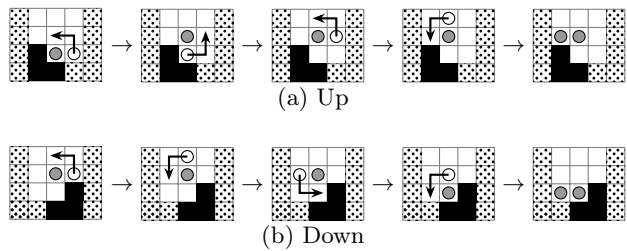

	\centering
		\begin{tabular}{c}

	\begin{customviewplot}{5}{4}

		\dontcarearea{0}{0}{0}{3}
		\dontcarearea{3}{0}{3}{0}
		\dontcarearea{4}{0}{4}{3}

		\gwall{1}{0}{2}{0}
		\gwall{1}{1}{1}{1}

		\gbackbone{2}{1}
		\gmodule{3}{1}

		\draw[myarrow] (3.5,1.5) -- (3.5,2.5) -- (2.2,2.5);

	\end{customviewplot}
	\ra
	\begin{customviewplot}{5}{4}

		\dontcarearea{0}{0}{0}{3}
		\dontcarearea{3}{0}{3}{0}
		\dontcarearea{4}{0}{4}{3}

		\gwall{1}{0}{2}{0}
		\gwall{1}{1}{1}{1}

		\gmodule{2}{1}
		\gbackbone{2}{2}

		\draw[myarrow] (2.5,1.5) -- (3.5,1.5) -- (3.5,2.8);

	\end{customviewplot}
	\ra
	\begin{customviewplot}{5}{4}

		\dontcarearea{0}{0}{0}{3}
		\dontcarearea{3}{0}{3}{0}
		\dontcarearea{4}{0}{4}{3}

		\gwall{1}{0}{2}{0}
		\gwall{1}{1}{1}{1}

		\gmodule{3}{2}
		\gbackbone{2}{2}

		\draw[myarrow] (3.5,2.5) -- (3.5,3.5) -- (2.2,3.5);

	\end{customviewplot}
	\ra
	\begin{customviewplot}{5}{4}

		\dontcarearea{0}{0}{0}{3}
		\dontcarearea{3}{0}{3}{0}
		\dontcarearea{4}{0}{4}{3}

		\gwall{1}{0}{2}{0}
		\gwall{1}{1}{1}{1}

		\gmodule{2}{3}
		\gbackbone{2}{2}

		\draw[myarrow] (2.5,3.5) -- (1.5,3.5) -- (1.5,2.2);

	\end{customviewplot}
	\ra
	\begin{customviewplot}{5}{4}

		\dontcarearea{0}{0}{0}{3}
		\dontcarearea{3}{0}{3}{0}
		\dontcarearea{4}{0}{4}{3}

		\gwall{1}{0}{2}{0}
		\gwall{1}{1}{1}{1}

		\gbackbone{2}{2}
		\gbackbone{1}{2}

	\end{customviewplot}

	\\
	\subcaption{(a) Up}\\
	\\

	\begin{customviewplot}{5}{4}

		\dontcarearea{0}{0}{0}{3}
		\dontcarearea{0}{0}{1}{0}
		\dontcarearea{4}{0}{4}{3}

		\gwall{2}{0}{3}{0}
		\gwall{3}{1}{3}{1}

		\gbackbone{2}{2}
		\gmodule{3}{2}

		\draw[myarrow] (3.5,2.5) -- (3.5,3.5) -- (2.2,3.5);

	\end{customviewplot}
	\ra
	\begin{customviewplot}{5}{4}

		\dontcarearea{0}{0}{0}{3}
		\dontcarearea{0}{0}{1}{0}
		\dontcarearea{4}{0}{4}{3}

		\gwall{2}{0}{3}{0}
		\gwall{3}{1}{3}{1}

		\gbackbone{2}{2}
		\gmodule{2}{3}

		\draw[myarrow] (2.5,3.5) -- (1.5,3.5) -- (1.5,2.2);

	\end{customviewplot}
	\ra
	\begin{customviewplot}{5}{4}

		\dontcarearea{0}{0}{0}{3}
		\dontcarearea{0}{0}{1}{0}
		\dontcarearea{4}{0}{4}{3}

		\gwall{2}{0}{3}{0}
		\gwall{3}{1}{3}{1}

		\gbackbone{2}{2}
		\gmodule{1}{2}

		\draw[myarrow] (1.5,2.5) -- (1.5,1.5) -- (2.8,1.5);

	\end{customviewplot}
	\ra
	\begin{customviewplot}{5}{4}

		\dontcarearea{0}{0}{0}{3}
		\dontcarearea{0}{0}{1}{0}
		\dontcarearea{4}{0}{4}{3}

		\gwall{2}{0}{3}{0}
		\gwall{3}{1}{3}{1}

		\gmodule{2}{2}
		\gbackbone{2}{1}

		\draw[myarrow] (2.5,2.5) -- (1.5,2.5) -- (1.5,1.2);

	\end{customviewplot}
	\ra
	\begin{customviewplot}{5}{4}

		\dontcarearea{0}{0}{0}{3}
		\dontcarearea{0}{0}{1}{0}
		\dontcarearea{4}{0}{4}{3}

		\gwall{2}{0}{3}{0}
		\gwall{3}{1}{3}{1}

		\gbackbone{2}{1}
		\gbackbone{1}{1}

	\end{customviewplot}
	\\
	\subcaption{(b) Down}\\

\end{tabular}

	\caption{Two modules go up/down a step}
	\label{fig:convex-moves-2-modules}
\end{figure}
\begin{figure}[t]
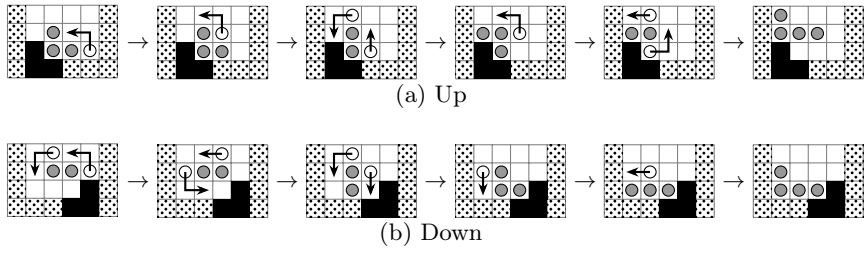

	\centering
		\begin{tabular}{c}

\begin{customviewplot}{6}{4}

\dontcarearea{0}{0}{0}{3}
\dontcarearea{3}{0}{4}{0}
\dontcarearea{5}{0}{5}{3}

\gwall{1}{0}{1}{1}
\gwall{2}{0}{2}{0}

\gbackbone{2}{2}
\gbackbone{2}{1}
\gbackbone{3}{1}
\gmodule{4}{1}

\draw[myarrow] (4.5,1.5) -- (4.5,2.5) -- (3.2,2.5);

\end{customviewplot}
\ra
\begin{customviewplot}{6}{4}

\dontcarearea{0}{0}{0}{3}
\dontcarearea{3}{0}{4}{0}
\dontcarearea{5}{0}{5}{3}

\gwall{1}{0}{1}{1}
\gwall{2}{0}{2}{0}

\gbackbone{2}{2}
\gbackbone{2}{1}
\gbackbone{3}{1}
\gmodule{3}{2}

\draw[myarrow] (3.5,2.5) -- (3.5,3.5) -- (2.2,3.5);

\end{customviewplot}
\ra
\begin{customviewplot}{6}{4}

\dontcarearea{0}{0}{0}{3}
\dontcarearea{3}{0}{4}{0}
\dontcarearea{5}{0}{5}{3}

\gwall{1}{0}{1}{1}
\gwall{2}{0}{2}{0}

\gbackbone{2}{2}
\gbackbone{2}{1}
\gmodule{3}{1}
\gmodule{2}{3}

\draw[myarrow] (2.5,3.5) -- (1.5,3.5) -- (1.5,2.2);
\draw[myarrow] (3.5,1.5) -- (3.5,2.8);

\end{customviewplot}
\ra
\begin{customviewplot}{6}{4}

\dontcarearea{0}{0}{0}{3}
\dontcarearea{3}{0}{4}{0}
\dontcarearea{5}{0}{5}{3}

\gwall{1}{0}{1}{1}
\gwall{2}{0}{2}{0}

\gbackbone{2}{2}
\gbackbone{2}{1}
\gmodule{3}{2}
\gbackbone{1}{2}

\draw[myarrow] (3.5,2.5) -- (3.5,3.5) -- (2.2,3.5);

\end{customviewplot}
\ra
\begin{customviewplot}{6}{4}

\dontcarearea{0}{0}{0}{3}
\dontcarearea{3}{0}{4}{0}
\dontcarearea{5}{0}{5}{3}

\gwall{1}{0}{1}{1}
\gwall{2}{0}{2}{0}

\gbackbone{2}{2}
\gmodule{2}{1}
\gmodule{2}{3}
\gbackbone{1}{2}

\draw[myarrow] (2.5,3.5) -- (1.2,3.5);
\draw[myarrow] (2.5,1.5) -- (3.5,1.5) -- (3.5,2.8);

\end{customviewplot}
\ra
\begin{customviewplot}{6}{4}

\dontcarearea{0}{0}{0}{3}
\dontcarearea{3}{0}{4}{0}
\dontcarearea{5}{0}{5}{3}

\gwall{1}{0}{1}{1}
\gwall{2}{0}{2}{0}

\gbackbone{1}{3}
\gbackbone{1}{2}
\gbackbone{2}{2}
\gbackbone{3}{2}

\end{customviewplot}

\\
\subcaption{(a) Up}\\
\\
\begin{customviewplot}{6}{4}

\dontcarearea{0}{0}{0}{3}
\dontcarearea{1}{0}{2}{0}
\dontcarearea{5}{0}{5}{3}

\gwall{3}{0}{4}{0}
\gwall{4}{1}{4}{1}

\gmodule{2}{3}
\gbackbone{2}{2}
\gbackbone{3}{2}
\gmodule{4}{2}

\draw[myarrow] (4.5,2.5) -- (4.5,3.5) -- (3.2,3.5);
\draw[myarrow] (2.5,3.5) -- (1.5,3.5) -- (1.5,2.2);

\end{customviewplot}
\ra
\begin{customviewplot}{6}{4}

\dontcarearea{0}{0}{0}{3}
\dontcarearea{1}{0}{2}{0}
\dontcarearea{5}{0}{5}{3}

\gwall{3}{0}{4}{0}
\gwall{4}{1}{4}{1}

\gmodule{1}{2}
\gbackbone{2}{2}
\gbackbone{3}{2}
\gmodule{3}{3}

\draw[myarrow] (1.5,2.5) -- (1.5,1.5) -- (2.8,1.5);
\draw[myarrow] (3.5,3.5) -- (2.2,3.5);

\end{customviewplot}
\ra
\begin{customviewplot}{6}{4}

\dontcarearea{0}{0}{0}{3}
\dontcarearea{1}{0}{2}{0}
\dontcarearea{5}{0}{5}{3}

\gwall{3}{0}{4}{0}
\gwall{4}{1}{4}{1}

\gbackbone{2}{1}
\gbackbone{2}{2}
\gmodule{3}{2}
\gmodule{2}{3}

\draw[myarrow] (2.5,3.5) -- (1.5,3.5) -- (1.5,2.2);
\draw[myarrow] (3.5,2.5) -- (3.5,1.2);

\end{customviewplot}
\ra
\begin{customviewplot}{6}{4}

\dontcarearea{0}{0}{0}{3}
\dontcarearea{1}{0}{2}{0}
\dontcarearea{5}{0}{5}{3}

\gwall{3}{0}{4}{0}
\gwall{4}{1}{4}{1}

\gbackbone{2}{1}
\gbackbone{2}{2}
\gbackbone{3}{1}
\gmodule{1}{2}

\draw[myarrow] (1.5,2.5) -- (1.5,1.2);

\end{customviewplot}
\ra
\begin{customviewplot}{6}{4}

\dontcarearea{0}{0}{0}{3}
\dontcarearea{1}{0}{2}{0}
\dontcarearea{5}{0}{5}{3}

\gwall{3}{0}{4}{0}
\gwall{4}{1}{4}{1}

\gbackbone{2}{1}
\gmodule{2}{2}
\gbackbone{3}{1}
\gbackbone{1}{1}

\draw[myarrow] (2.5,2.5) -- (1.2,2.5);

\end{customviewplot}
\ra
\begin{customviewplot}{6}{4}

\dontcarearea{0}{0}{0}{3}
\dontcarearea{1}{0}{2}{0}
\dontcarearea{5}{0}{5}{3}

\gwall{3}{0}{4}{0}
\gwall{4}{1}{4}{1}

\gbackbone{2}{1}
\gbackbone{1}{2}
\gbackbone{3}{1}
\gbackbone{1}{1}

\end{customviewplot}

\\
\subcaption{(b) Down}\\

\end{tabular}

	\caption{Four modules go up/down a step}
	\label{fig:convex-moves-4-modules}
\end{figure}
\begin{figure}[t]
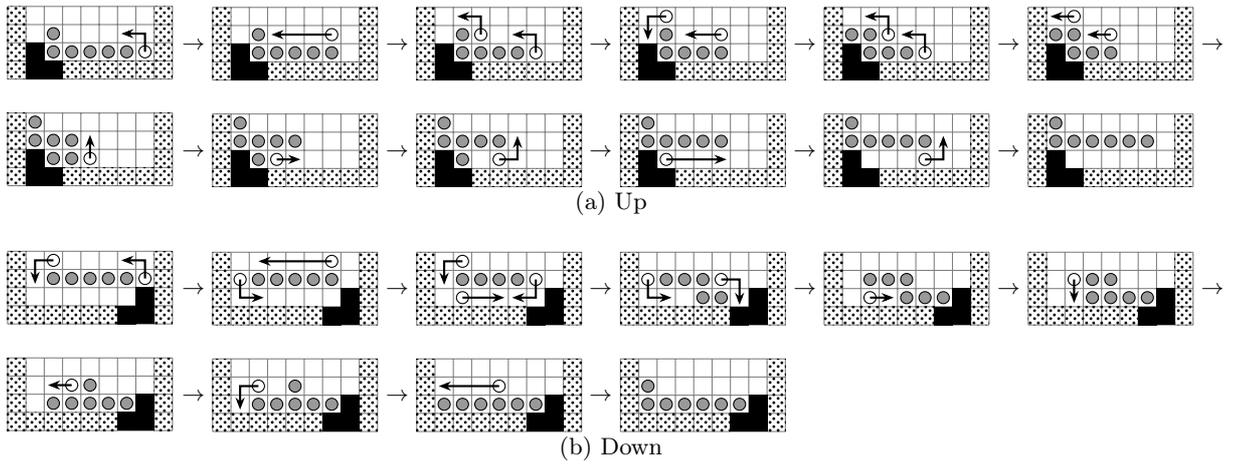

	\centering
		\begin{tabular}{c}

\begin{tabular}{l}
\begin{customviewplot}{9}{4}

\dontcarearea{0}{0}{0}{3}
\dontcarearea{3}{0}{7}{0}
\dontcarearea{8}{0}{8}{3}

\gwall{1}{0}{1}{1}
\gwall{2}{0}{2}{0}

\gbackbone{2}{2}
\gbackbone{2}{1}
\gbackbone{3}{1}
\gbackbone{4}{1}
\gbackbone{5}{1}
\gbackbone{6}{1}
\gmodule{7}{1}

\draw[myarrow] (7.5,1.5) -- (7.5,2.5) -- (6.2,2.5);

\end{customviewplot}
\ra
\begin{customviewplot}{9}{4}

\dontcarearea{0}{0}{0}{3}
\dontcarearea{3}{0}{7}{0}
\dontcarearea{8}{0}{8}{3}

\gwall{1}{0}{1}{1}
\gwall{2}{0}{2}{0}

\gbackbone{2}{2}
\gbackbone{2}{1}
\gbackbone{3}{1}
\gbackbone{4}{1}
\gbackbone{5}{1}
\gbackbone{6}{1}
\gmodule{6}{2}

\draw[myarrow] (6.5,2.5) -- (3.2,2.5);

\end{customviewplot}
\ra
\begin{customviewplot}{9}{4}

\dontcarearea{0}{0}{0}{3}
\dontcarearea{3}{0}{7}{0}
\dontcarearea{8}{0}{8}{3}

\gwall{1}{0}{1}{1}
\gwall{2}{0}{2}{0}

\gbackbone{2}{2}
\gbackbone{2}{1}
\gbackbone{3}{1}
\gbackbone{4}{1}
\gbackbone{5}{1}
\gmodule{6}{1}
\gmodule{3}{2}

\draw[myarrow] (6.5,1.5) -- (6.5,2.5) -- (5.2,2.5);
\draw[myarrow] (3.5,2.5) -- (3.5,3.5) -- (2.2,3.5);

\end{customviewplot}
\ra
\begin{customviewplot}{9}{4}

\dontcarearea{0}{0}{0}{3}
\dontcarearea{3}{0}{7}{0}
\dontcarearea{8}{0}{8}{3}

\gwall{1}{0}{1}{1}
\gwall{2}{0}{2}{0}

\gbackbone{2}{2}
\gbackbone{2}{1}
\gbackbone{3}{1}
\gbackbone{4}{1}
\gbackbone{5}{1}
\gmodule{5}{2}
\gmodule{2}{3}

\draw[myarrow] (2.5,3.5) -- (1.5,3.5) -- (1.5,2.2);
\draw[myarrow] (5.5,2.5) -- (3.5,2.5);

\end{customviewplot}
\ra
\begin{customviewplot}{9}{4}

\dontcarearea{0}{0}{0}{3}
\dontcarearea{3}{0}{7}{0}
\dontcarearea{8}{0}{8}{3}

\gwall{1}{0}{1}{1}
\gwall{2}{0}{2}{0}

\gbackbone{2}{2}
\gbackbone{2}{1}
\gbackbone{3}{1}
\gbackbone{4}{1}
\gmodule{5}{1}
\gmodule{3}{2}
\gbackbone{1}{2}

\draw[myarrow] (3.5,2.5) -- (3.5,3.5) -- (2.2,3.5);
\draw[myarrow] (5.5,1.5) -- (5.5,2.5) -- (4.2,2.5);

\end{customviewplot}
\ra
\begin{customviewplot}{9}{4}

\dontcarearea{0}{0}{0}{3}
\dontcarearea{3}{0}{7}{0}
\dontcarearea{8}{0}{8}{3}

\gwall{1}{0}{1}{1}
\gwall{2}{0}{2}{0}

\gbackbone{2}{2}
\gbackbone{2}{1}
\gbackbone{3}{1}
\gbackbone{4}{1}
\gmodule{4}{2}
\gmodule{2}{3}
\gbackbone{1}{2}

\draw[myarrow] (2.5,3.5) -- (1.2,3.5);
\draw[myarrow] (4.5,2.5) -- (3.2,2.5);

\end{customviewplot}
\ra
\\ \\
\begin{customviewplot}{9}{4}

\dontcarearea{0}{0}{0}{3}
\dontcarearea{3}{0}{7}{0}
\dontcarearea{8}{0}{8}{3}

\gwall{1}{0}{1}{1}
\gwall{2}{0}{2}{0}

\gbackbone{2}{2}
\gbackbone{2}{1}
\gbackbone{3}{1}
\gmodule{4}{1}
\gbackbone{3}{2}
\gbackbone{1}{3}
\gbackbone{1}{2}

\draw[myarrow] (4.5,1.5) -- (4.5,2.8);

\end{customviewplot}
\ra
\begin{customviewplot}{9}{4}

\dontcarearea{0}{0}{0}{3}
\dontcarearea{3}{0}{7}{0}
\dontcarearea{8}{0}{8}{3}

\gwall{1}{0}{1}{1}
\gwall{2}{0}{2}{0}

\gbackbone{2}{2}
\gbackbone{2}{1}
\gmodule{3}{1}
\gbackbone{4}{2}
\gbackbone{3}{2}
\gbackbone{1}{3}
\gbackbone{1}{2}

\draw[myarrow] (3.5,1.5) -- (4.8,1.5);

\end{customviewplot}
\ra
\begin{customviewplot}{9}{4}

\dontcarearea{0}{0}{0}{3}
\dontcarearea{3}{0}{7}{0}
\dontcarearea{8}{0}{8}{3}

\gwall{1}{0}{1}{1}
\gwall{2}{0}{2}{0}

\gbackbone{2}{2}
\gbackbone{2}{1}
\gmodule{4}{1}
\gbackbone{4}{2}
\gbackbone{3}{2}
\gbackbone{1}{3}
\gbackbone{1}{2}

\draw[myarrow] (4.5,1.5) -- (5.5,1.5) -- (5.5,2.8);

\end{customviewplot}
\ra
\begin{customviewplot}{9}{4}

\dontcarearea{0}{0}{0}{3}
\dontcarearea{3}{0}{7}{0}
\dontcarearea{8}{0}{8}{3}

\gwall{1}{0}{1}{1}
\gwall{2}{0}{2}{0}

\gbackbone{2}{2}
\gmodule{2}{1}
\gbackbone{5}{2}
\gbackbone{4}{2}
\gbackbone{3}{2}
\gbackbone{1}{3}
\gbackbone{1}{2}

\draw[myarrow] (2.5,1.5) -- (5.8,1.5);

\end{customviewplot}
\ra
\begin{customviewplot}{9}{4}

\dontcarearea{0}{0}{0}{3}
\dontcarearea{3}{0}{7}{0}
\dontcarearea{8}{0}{8}{3}

\gwall{1}{0}{1}{1}
\gwall{2}{0}{2}{0}

\gbackbone{2}{2}
\gmodule{5}{1}
\gbackbone{5}{2}
\gbackbone{4}{2}
\gbackbone{3}{2}
\gbackbone{1}{3}
\gbackbone{1}{2}

\draw[myarrow] (5.5,1.5) -- (6.5,1.5) -- (6.5,2.8);

\end{customviewplot}
\ra
\begin{customviewplot}{9}{4}

\dontcarearea{0}{0}{0}{3}
\dontcarearea{3}{0}{7}{0}
\dontcarearea{8}{0}{8}{3}

\gwall{1}{0}{1}{1}
\gwall{2}{0}{2}{0}

\gbackbone{2}{2}
\gbackbone{6}{2}
\gbackbone{5}{2}
\gbackbone{4}{2}
\gbackbone{3}{2}
\gbackbone{1}{3}
\gbackbone{1}{2}

\end{customviewplot}

\end{tabular}
\\
\subcaption{(a) Up}\\
\\

\begin{tabular}{l}

\begin{customviewplot}{9}{4}

\dontcarearea{0}{0}{0}{3}
\dontcarearea{1}{0}{5}{0}
\dontcarearea{8}{0}{8}{3}

\gwall{7}{0}{7}{1}
\gwall{6}{0}{6}{0}

\gmodule{2}{3}
\gbackbone{2}{2}
\gbackbone{3}{2}
\gbackbone{4}{2}
\gbackbone{5}{2}
\gbackbone{6}{2}
\gmodule{7}{2}

\draw[myarrow] (7.5,2.5) -- (7.5,3.5) -- (6.2,3.5);
\draw[myarrow] (2.5,3.5) -- (1.5,3.5) -- (1.5,2.2);

\end{customviewplot}
\ra
\begin{customviewplot}{9}{4}

\dontcarearea{0}{0}{0}{3}
\dontcarearea{1}{0}{5}{0}
\dontcarearea{8}{0}{8}{3}

\gwall{7}{0}{7}{1}
\gwall{6}{0}{6}{0}

\gmodule{1}{2}
\gbackbone{2}{2}
\gbackbone{3}{2}
\gbackbone{4}{2}
\gbackbone{5}{2}
\gbackbone{6}{2}
\gmodule{6}{3}

\draw[myarrow] (1.5,2.5) -- (1.5,1.5) -- (2.8,1.5);
\draw[myarrow] (6.5,3.5) -- (2.5,3.5);

\end{customviewplot}
\ra
\begin{customviewplot}{9}{4}

\dontcarearea{0}{0}{0}{3}
\dontcarearea{1}{0}{5}{0}
\dontcarearea{8}{0}{8}{3}

\gwall{7}{0}{7}{1}
\gwall{6}{0}{6}{0}

\gmodule{2}{1}
\gbackbone{2}{2}
\gbackbone{3}{2}
\gbackbone{4}{2}
\gbackbone{5}{2}
\gmodule{6}{2}
\gmodule{2}{3}

\draw[myarrow] (2.5,3.5) -- (1.5,3.5) -- (1.5,2.2);
\draw[myarrow] (2.5,1.5) -- (4.8,1.5);
\draw[myarrow] (6.5,2.5) -- (6.5,1.5) -- (5.2,1.5);

\end{customviewplot}
\ra
\begin{customviewplot}{9}{4}

\dontcarearea{0}{0}{0}{3}
\dontcarearea{1}{0}{5}{0}
\dontcarearea{8}{0}{8}{3}

\gwall{7}{0}{7}{1}
\gwall{6}{0}{6}{0}

\gbackbone{2}{2}
\gbackbone{3}{2}
\gbackbone{4}{2}
\gmodule{5}{2}
\gbackbone{4}{1}
\gbackbone{5}{1}
\gmodule{1}{2}

\draw[myarrow] (5.5,2.5) -- (6.5,2.5) -- (6.5,1.1);
\draw[myarrow] (1.5,2.5) -- (1.5,1.5) -- (2.8,1.5);

\end{customviewplot}
\ra
\begin{customviewplot}{9}{4}

\dontcarearea{0}{0}{0}{3}
\dontcarearea{1}{0}{5}{0}
\dontcarearea{8}{0}{8}{3}

\gwall{7}{0}{7}{1}
\gwall{6}{0}{6}{0}

\gbackbone{2}{2}
\gbackbone{3}{2}
\gbackbone{4}{2}
\gbackbone{4}{1}
\gbackbone{5}{1}
\gbackbone{6}{1}
\gmodule{2}{1}

\draw[myarrow] (2.5,1.5) -- (3.8,1.5);

\end{customviewplot}
\ra
\begin{customviewplot}{9}{4}

\dontcarearea{0}{0}{0}{3}
\dontcarearea{1}{0}{5}{0}
\dontcarearea{8}{0}{8}{3}

\gwall{7}{0}{7}{1}
\gwall{6}{0}{6}{0}

\gmodule{2}{2}
\gbackbone{3}{2}
\gbackbone{4}{2}
\gbackbone{3}{1}
\gbackbone{4}{1}
\gbackbone{5}{1}
\gbackbone{6}{1}

\draw[myarrow] (2.5,2.5) -- (2.5,1.2);

\end{customviewplot}
\ra
\\ \\
\begin{customviewplot}{9}{4}

\dontcarearea{0}{0}{0}{3}
\dontcarearea{1}{0}{5}{0}
\dontcarearea{8}{0}{8}{3}

\gwall{7}{0}{7}{1}
\gwall{6}{0}{6}{0}

\gmodule{3}{2}
\gbackbone{4}{2}
\gbackbone{2}{1}
\gbackbone{3}{1}
\gbackbone{4}{1}
\gbackbone{5}{1}
\gbackbone{6}{1}

\draw[myarrow] (3.5,2.5) -- (2.2,2.5);

\end{customviewplot}
\ra
\begin{customviewplot}{9}{4}

\dontcarearea{0}{0}{0}{3}
\dontcarearea{1}{0}{5}{0}
\dontcarearea{8}{0}{8}{3}

\gwall{7}{0}{7}{1}
\gwall{6}{0}{6}{0}

\gmodule{2}{2}
\gbackbone{4}{2}
\gbackbone{2}{1}
\gbackbone{3}{1}
\gbackbone{4}{1}
\gbackbone{5}{1}
\gbackbone{6}{1}

\draw[myarrow] (2.5,2.5) -- (1.5,2.5) -- (1.5,1.2);

\end{customviewplot}
\ra
\begin{customviewplot}{9}{4}

\dontcarearea{0}{0}{0}{3}
\dontcarearea{1}{0}{5}{0}
\dontcarearea{8}{0}{8}{3}

\gwall{7}{0}{7}{1}
\gwall{6}{0}{6}{0}

\gmodule{4}{2}
\gbackbone{1}{1}
\gbackbone{2}{1}
\gbackbone{3}{1}
\gbackbone{4}{1}
\gbackbone{5}{1}
\gbackbone{6}{1}

\draw[myarrow] (4.5,2.5) -- (1.2,2.5);

\end{customviewplot}
\ra
\begin{customviewplot}{9}{4}

\dontcarearea{0}{0}{0}{3}
\dontcarearea{1}{0}{5}{0}
\dontcarearea{8}{0}{8}{3}

\gwall{7}{0}{7}{1}
\gwall{6}{0}{6}{0}

\gbackbone{1}{2}
\gbackbone{1}{1}
\gbackbone{2}{1}
\gbackbone{3}{1}
\gbackbone{4}{1}
\gbackbone{5}{1}
\gbackbone{6}{1}

\end{customviewplot}

\end{tabular}

\\
\subcaption{(b) Down}\\

\end{tabular}

	\caption{Seven modules go up/down a step}
	\label{fig:convex-moves-7-modules}
\end{figure}

\begin{figure}[t]
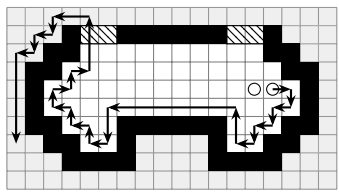

	\centering
	\begin{customviewplot}{18}{10}
\outside{1}
\outsidearea{1}{1}{1}{2}
\outsidearea{2}{1}{2}{1}

\outsidearea{1}{7}{1}{8}
\outsidearea{2}{8}{2}{8}

\outsidearea{15}{1}{16}{1}
\outsidearea{16}{2}{16}{2}

\outsidearea{15}{8}{16}{8}
\outsidearea{16}{7}{16}{7}

\outsidearea{7}{1}{10}{2}

\gwall{3}{1}{5}{1}
\gwall{12}{1}{14}{1}

\gwall{6}{1}{6}{3}
\gwall{11}{1}{11}{3}
\gwall{6}{3}{11}{3}

\gwall{1}{3}{1}{6}
\gwall{3}{8}{3}{8}
\gwall{2}{2}{2}{3}
\gwall{2}{2}{3}{2}
\gwall{2}{6}{2}{7}
\gwall{2}{7}{3}{7}

\gwall{16}{3}{16}{6}
\gwall{14}{8}{14}{8}
\gwall{14}{2}{15}{2}
\gwall{15}{2}{15}{3}
\gwall{14}{7}{15}{7}
\gwall{15}{6}{15}{7}

\gwall{6}{8}{11}{8}
\gexit{4}{8}{5}{8}
\gexit{12}{8}{13}{8}

\gmodule{13}{5}
\gmodule{14}{5}

\draw[myarrow] (14.5,5.5) -- (15.5,5.5);
\draw[myarrow] (15.5,5.5) -- (15.5,4.5);
\draw[myarrow] (15.5,4.5) -- (14.5,4.5);
\draw[myarrow] (14.5,4.5) -- (14.5,3.5);
\draw[myarrow] (14.5,3.5) -- (13.5,3.5);
\draw[myarrow] (13.5,3.5) -- (13.5,2.5);
\draw[myarrow] (13.5,2.5) -- (12.5,2.5);
\draw[myarrow] (12.5,2.5) -- (12.5,4.5);
\draw[myarrow] (12.5,4.5) -- (5.5,4.5);
\draw[myarrow] (5.5,4.5) -- (5.5,2.5);
\draw[myarrow] (5.5,2.5) -- (4.5,2.5);
\draw[myarrow] (4.5,2.5) -- (4.5,3.5);
\draw[myarrow] (4.5,3.5) -- (3.5,3.5);
\draw[myarrow] (3.5,3.5) -- (3.5,4.5);
\draw[myarrow] (3.5,4.5) -- (2.5,4.5);
\draw[myarrow] (2.5,4.5) -- (2.5,5.5);
\draw[myarrow] (2.5,5.5) -- (3.5,5.5);
\draw[myarrow] (3.5,5.5) -- (3.5,6.5);
\draw[myarrow] (3.5,6.5) -- (4.5,6.5);
\draw[myarrow] (4.5,6.5) -- (4.5,9.5);
\draw[myarrow] (4.5,9.5) -- (2.5,9.5);
\draw[myarrow] (2.5,9.5) -- (2.5,8.5);
\draw[myarrow] (2.5,8.5) -- (1.5,8.5);
\draw[myarrow] (1.5,8.5) -- (1.5,7.5);
\draw[myarrow] (1.5,7.5) -- (0.5,7.5);
\draw[myarrow] (0.5,7.5) -- (0.5,2.5);

\end{customviewplot}
	\caption{An example evacuation path of an MRS from a convex field}
	\label{fig:mrs-trail-in-convex-field}
\end{figure}

\section{Conclusion}
\label{sec:conclusion}

In this paper, we considered evacuation from a finite 2D square grid field by a metamorphic robotic system (MRS).
We focused on the minimum number of modules required to solve the evacuation problem under several conditions.
First, we considered a rectangular field surrounded by walls with at least one exit and showed that two modules are necessary and sufficient for any rectangular field if the modules are equipped with a global compass.
Then, we focused on cases where modules do not have a global compass and showed that four (resp.~seven) modules are necessary and sufficient for restricted (resp.~any) initial states of an MRS.
Additionally, we showed that two modules are sufficient if an MRS is touching a wall in its initial position.
We also clarified the condition to stop an MRS from moving after evacuation of a rectangular field.
Finally, we extended these results to mazes and convex fields.

Various open problems exist regarding evacuation by an MRS.
In Section \ref{sec:stop-after-evacuation}, we concluded that it is impossible for an MRS to stop moving after evacuation from a maze, even if we restrict the initial shapes of an MRS.
From this fact, we have a natural question;
what is a necessary function to stop after evacuation from a maze?
While we demonstrated that the proposed evacuation algorithms in Section~\ref{sec:evacuation-rectangular} can be extended for a maze, it is not known whether every evacuation algorithm for a rectangular field can be extended to a maze.
It would be interesting to consider whether this is possible.
It is also unclear whether the assumptions we made for the proposed algorithms, i.e., the visibility range and common handedness, are optimal.
Another open problem is finding an evacuation algorithm for two or more MRSs, while avoiding a collision between them.

\section*{Acknowledgements}
This work was supported by JSPS KAKENHI Grant Numbers 19K11828 and 20KK0232, and Israel \& Japan Science and Technology Agency (JST) SICORP (Grant\#JPMJSC1806).

\bibliographystyle{WileyNJD-AMA}
\bibliography{collection}

\end{document}